\documentclass[letterpaper,12pt,notitlepage]{article}
 
 \usepackage{amssymb}
 \usepackage{amsfonts}
 \usepackage{amsmath}
 \usepackage{bbm}
 \usepackage{ bbold }
 
 \usepackage{setspace} 
 \usepackage[margin=1.1in]{geometry}
 \usepackage[utf8]{inputenc}
 \usepackage{enumerate}
 \usepackage{amsmath}
 \usepackage{amsthm}
 \usepackage{MnSymbol,wasysym}

  \usepackage{threeparttable}
  \usepackage{makecell}
  \usepackage{booktabs}
 
 \usepackage{threeparttable}
 \usepackage{graphicx} 
 \usepackage[svgnames]{xcolor}
 
 \usepackage{amsfonts}
 \providecommand{\abs}[1]{\lvert#1\rvert}

 \usepackage{verbatim}
 \providecommand{\R}{\mathbb{R}}
 \providecommand{\R}{\mathbb{R}}
 \providecommand{\E}{\mathbb{E}}

 \usepackage[round]{natbib}
 \usepackage[hidelinks]{hyperref}
 \hypersetup{
 	colorlinks=false,
 	linkcolor=blue,
 	filecolor=magenta,      
 	urlcolor=Blue,
 	citecolor=Blue,
 	pdftitle={Sharelatex Example},
 	bookmarks=true
 }

\begin{document}
\title{Confidence and Organizations\thanks{I am grateful to David Besanko, Luis Rayo, Daniel Barron, and Wojciech Olszewski for their guidance and encouragement. I would also like to thank Sarah Auster, Henrique Castro-Pires, Francesc Dilme, Xavier Durán, Théo Durandard, Alkis Georgiadis-Harris, Benjamin Golub, Botond K\H{o}szegi, David Laibson, Ulrike Malmendier, Julien Manili, Andrei Matveenko, Edwin Mu\~{n}oz-Rodr\'{i}guez, Michael Powell, Rosina Rodr\'{i}guez-Olivera, Bruno Strulovici, Alvaro Sandroni, Jeroen Swinkels, Boli Xu, Gabriel Ziegler, as well as participants at various seminars and conferences for their insightful comments and suggestions. Support by the Deutsche Forschungsgemeinschaft (German Research Foundation) through grant CRC TR 224 (Project B02) is gratefully acknowledged.}}
\author{Andrés Espitia\thanks{Kellogg School of Management - Northwestern University. Email: \href{mailto:andres.espitia@kellogg.northwestern.edu}{andres.espitia@kellogg.northwestern.edu}.}}
\maketitle
\newtheorem{assumption}{Assumption}
\newtheorem{lemma}{Lemma}
\newtheorem{claim}{Claim}
\newtheorem{proposition}{Proposition}
\newtheorem{corollary}{Corollary}
\newtheorem{definition}{Definition}
\vspace{-3em}

\begin{abstract}
	Miscalibrated beliefs are widely viewed as compromising the quality of employees’ decisions. Why, then, might an organization prefer to hire an individual known to be overconfident? This paper develops a theory of organizational demand for employees' levels of confidence when private information interacts with conflicts of interest. I study a model in which an employee uses private information to make decisions on behalf of the organization and analyze the \textit{belief design problem}, namely, how the organization would like the employee to interpret his observations. I show that organizations prefer employees whose actions reflect a constant expected conflict of interest across observations. A well-calibrated employee is optimal if and only if private information does not affect this conflict. When the conflict varies with information, organizations optimally select employees whose confidence distorts their responses to information. Overconfidence is optimal when the organization seeks stronger adjustments to information than a well-calibrated employee would provide.
\end{abstract}



\newpage

\section{Introduction}

There is extensive evidence about the pervasiveness of overconfidence in organizations \citep{malmendier2015verges}.  CEOs, executives, and managers \citep{malmendier2005ceo,ben2013managerial,adam2015managerial, barrero2021micro}; traders and investors \citep{daniel2015overconfident}; lawyers \citep{goodman2010insightful}; medical doctors \citep{berner2008overconfidence,croskerry2008overconfidence}; and entrepreneurs \citep{koellinger2007think} have all been found to exhibit some degree of overconfidence.\footnote{Throughout the paper, I use “overconfidence” as a generic label. I focus on overprecision, whose informal meaning is discussed in the Introduction and whose formalization is provided in \autoref{sec:appsanddefs}.} This behavioral bias is regularly considered as compromising the quality of employees’ decisions and as harming organizations’ performance.\footnote{In the words of Nobel laureate Daniel Kahneman, “an unbiased appreciation of uncertainty is a cornerstone of rationality – but it is not what people and organizations want” \citep[see][chap.~24]{kahneman2011thinking}. Furthermore, “Kahneman recently told an interviewer that if he had a magic wand that could eliminate one human bias, he would do away with overconfidence” \citep[pg.~1]{malmendier2015verges}.} Given the costs that organizations suffer due to overconfidence, why might employees \textit{known} to be overconfident be systematically hired and retained, even when their unbiased counterparts are available? Additionally, how would organizations use their knowledge about applicants’ confidence to select their employees optimally?

Motivated by this puzzle, I explore the role of biased beliefs as an instrument to alleviate \textit{agency frictions}, those arising from misaligned preferences among members of an organization. I focus on situations in which an employee is entrusted with a decision and relies on private information to act on behalf of the organization. In such settings, the key friction is often not the lack of information but the responsiveness of the employee’s judgment to the information he observes.

More broadly, in many environments, different members of the organization agree on basic empirical regularities but disagree on how informative particular observations are. Base rates of underlying conditions and the frequency with which observable signals occur are often public or can be learned quickly from experience. What is harder to validate is how strongly a given signal realization should shift beliefs about the underlying state.

Accordingly, I study environments in which agents agree on the distribution of states and on the frequency of signals, but may disagree on how to combine these two pieces of information. One can think of agents as sharing the same empirical “spreadsheets” for states and signals, yet differing in how they merge them into an inference rule. This disagreement captures differences in confidence or perceived informativeness, while holding other belief distortions, such as optimism about fundamentals, fixed.

This framework allows me to study a particular manifestation of overconfidence: \textit{overprecision}, an exaggerated faith in one’s information. Informally, overprecise individuals tend to regard themselves as more informed than is justified by reality.\footnote{Other forms of overconfidence include overestimation (thinking that one’s performance and abilities are above their actual level) and overplacement (erroneously thinking that one has outperformed others or that one’s abilities are above those of other individuals). See \citet{moore2008trouble} for a discussion on the connections and differences between these manifestations of overconfidence.} Intuitively, an overconfident employee is desirable when a well-calibrated counterpart would be too \textit{unresponsive} to new information. Overprecision leads to over-updating and to more extreme actions, which may benefit the organization when calibrated employees would otherwise remain too close to a reference action.

Moreover, different positions within organizations place different demands on how responsive decisions should be to information. Some roles require frequent adjustment to changing external conditions, while others emphasize consistency and adherence to established procedures. The analysis captures these differences in reduced form and studies how they translate into organizational demand for employees with different worldviews, or levels of confidence.

To examine the consequences of this mechanism, I develop a model in which an employee makes a decision that affects both himself and the organization. While there is disagreement about the right course of action, the employee has private access to relevant but unverifiable information. In this setting, I introduce a novel feature: the organization can choose how the employee interprets his observations. I refer to this decision as the \textit{belief design problem}. This paper is best understood as a theory of how workplace positions generate demand for particular employee worldviews, using belief design as an analytical tool.

Formally, belief design corresponds to an optimization problem in which the organization chooses the employee’s posterior beliefs after each signal realization, subject to the constraint that they both agree \textit{ex ante} on the distribution of the state and signals. In practice, belief design can be interpreted as a selection process in which the organization chooses among candidates with different worldviews. This differs from standard information design \citep[see][]{kamenica2019bayesian}, which focuses on the provision of information rather than on how information is interpreted.

The contributions of this paper are threefold. First, I characterize conditions under which the organization prefers a well-calibrated, underconfident, or overconfident employee. Second, I introduce belief design as a flexible, tractable methodology for studying belief-based biases. A key step is to interpret the beliefs optimally chosen by the organization. I propose a definition of overprecision based on the concordance stochastic order that applies to a general class of information structures, going beyond commonly assumed parametric restrictions such as the bivariate normal or truth-or-noise structures. The latter case is studied in \autoref{sec:truth}, where the same economic forces are shown to operate in a continuous-state environment. Third, I study how belief-based selection interacts with other organizational practices used to mitigate agency frictions, including action-contingent transfers and centralized decision-making.

A key object for describing the optimal employee’s characteristics is the difference between the actions preferred by each player in each state, which I refer to as the \textit{conflict of interest}. Misalignment in preferences generates an average bias in the employee's actions. Since there is \textit{ex ante} agreement on the marginal distributions of states and signals, belief-based selection allows this fixed average bias to be redistributed across signal realizations. Under the standard assumption of quadratic-loss preferences, the organization dislikes changes in the employee’s actions that are not justified by his private information. I show that the optimal employee aligns with the organization in terms of \textit{responsiveness} to information, in the sense that the bias in his actions does not change with signal realizations. As a result, a well-calibrated employee is preferred if and only if the expected conflict of interest is invariant across signal realizations.

In contrast, if the signal affects the expected conflict of interest, evenly distributing the employee’s bias requires that the optimal employee takes lower actions than his well-calibrated counterpart after signal realizations that induce high expected conflict of interest. The critical condition for the optimality of overconfidence is that the signal moves the conditional expectation of the conflict of interest and the employee’s preferred action in \textit{opposite} directions. If this is the case, the optimal agent takes more extreme actions than the well-calibrated one, which is a manifestation of overconfidence. Analogously, an underconfident employee is optimal if the expected conflict of interest and the employee’s preferred action move in the \textit{same} direction.

Optimality of overconfidence arises, for example, when the employee’s preferred action increases less than proportionally relative to the organization’s preferred action. This may occur because it is costly for the employee to adjust to the current conditions or because he is subject to some degree of \textit{status quo} bias.\footnote{It is sometimes assumed that adjustment costs are borne by the firm  \citep[see][]{barrero2021micro}. In the benchmark case with no state-dependent conflict of interest, including the special case of fully aligned preferences, Proposition 2 implies that belief distortions are detrimental for the firm. The potential benefits of biased beliefs arise only when conflicts vary across states.} Thus, the demand for overconfidence may arise as a strategy to mitigate the effects of other pervasive behavioral biases, \textit{status quo} bias in decision-making being a salient example.

I also explore the effects of alternative tools organizations may use to alleviate agency frictions. First, I study the interaction between belief design and action-contingent transfers.  I discuss conditions under which the use of transfers does not change the optimal beliefs. Moreover, in the optimal case, these tools serve different purposes: beliefs are used to spread the employee’s bias across signal realizations, while transfers are used to decrease his average bias. Interestingly, belief design leads to ``flatter" compensation contracts: all equilibrium actions taken by the optimal employee yield the same transfer. In contrast, when the employee is restricted to be well-calibrated, the optimal transfers vary with the actions he takes in equilibrium.

Finally, although modeled symmetrically in this paper, overconfidence and underconfidence differ starkly: there is a natural substitute for extreme underconfidence. Namely, the organization can retain decision rights and act without relying on the employee’s private information. Delegation is valuable only insofar as it allows the organization to use that information. When the optimal employee is insufficiently confident, the organization is therefore better off avoiding the conflict of interest altogether by centralizing decision-making.

\textbf{Related literature.}
The economic analysis of overconfidence has a long tradition.\footnote{As early as Adam Smith, it was noted that “the over-weening conceit which the greater part of men have of their own abilities, is an ancient evil remarked by the philosophers and moralists of all ages” \citep[chap.~X, book~I]{adamsmith}.} The term, however, encompasses distinct phenomena \citep{moore2008trouble}. This paper focuses on \textit{overprecision}, understood as an exaggerated faith in the informativeness of one’s signals, which is often regarded as the most robust and least well-understood form of overconfidence \citep{haran2010simple,moore2017three}. For a comprehensive review of overconfidence in labor markets and organizations, see \citet{santos2020overconfidence}.

A growing body of theoretical literature has examined settings in which overconfidence (broadly defined) may improve organizational outcomes. At a high level, these arguments fall into three categories. First, overconfidence may serve as a commitment device, enabling organizations or their members to pursue strategies that would otherwise lack credibility \citep{kyle1997speculation,rotemberg2000visionaries,van2005organizational,gervais2007positive,englmaier2011commitment,bolton2013leadership,englmaier2014biased,phua2018overconfident,ba2023multi}. Second, miscalibrated beliefs may facilitate the acquisition, revelation, or aggregation of private information \citep{bernardo2001evolution,vidal2007should,che2009opinions,van2010culture,levy2015correlation,hestermann2020experimentation,ilinov2022optimally,ostrizekvague}. Third, overconfidence may improve risk-taking, risk-sharing, or diversification within organizations \citep{goel2008overconfidence,santos2008positive,de2011overconfidence,gervais2011overconfidence,palomino2011overconfidence,heller2014overconfidence,lambertsen2025exploiting}.

This paper differs from these approaches by linking the optimality of overconfidence directly to the structure of preference misalignment within the organization. Rather than relying on specific strategic interactions or informational externalities, the analysis isolates how the interaction between private information and state-dependent conflicts of interest generates demand for particular belief distortions.

While much of the literature treats belief distortions as exogenous, there is increasing interest in models that endogenize belief formation. Work on persuasion and narratives, where agents influence how others map data into beliefs, is particularly relevant \citep{eliaz2020model,schwartzstein2021using,jain2023informing,aina2024tailored,ispano2023perils}. In the present context, the organization’s choice of how the employee interprets his private information can be viewed as a choice of narrative. Unlike most of this literature, belief design is constrained by \textit{ex ante} agreement on the marginal distributions of states and signals, and the analysis provides an interpretation of the optimal belief distortion in terms of overprecision.

Finally, I illustrate how belief-based selection can be used as an indirect source of incentives. Numerous tools for addressing conflicting preferences have been studied. One possibility is to exploit individual characteristics to improve outcomes \citep{prendergast1996impetuous,prendergast2007motivation,prendergast2008intrinsic}. I contribute to this literature by considering beliefs as part of those characteristics. The literature on information design \citep{rayo2010optimal,kamenica2011bayesian,kamenica2019bayesian} shares the same object of choice, but focuses on the information that employees actually observe rather than on how they interpret exogenously available information. Alternatively, the delegation literature \citep{delegation1977,delegation1984,alonso2008optimal} studies the use of rules on the set of available actions from which the employee can choose. I build on a similar framework, namely an organization formed by two individuals with different preferences over possible decisions and a mismatch between authority and information, and I abstract from moral hazard and contractual incentives in order to isolate how belief-based selection shapes delegated decision-making.

\section{Model}
\label{sec:model}

\qquad \textbf{Preliminaries. }
For any finite set $X\subset\R$,  I use $\Delta(X)$ to denote the set of probability mass functions over $X$.  Subscripts on operators explicitly specify the probability mass function used; e.g., $\mathbb{E}g$ indicates that the expectation is taken with respect to the distribution $g\in\Delta(X)$. For any joint distribution $g\in\Delta(X\times X’)$, let $g_X\in\Delta(X)$ and  $g_{X’}\in\Delta(X’)$ denote its marginals. Finally, all variables with a tilde are random variables.

\textbf{Players and actions. }
An agent (\textit{he}) makes a decision $x\in\R$ that also affects a principal (\textit{she}). Payoffs depend on a state of the world $\theta\in\Theta:=\{\theta_1,\ldots,\theta_n \}\subset \R$, with $n\geq 2$. States are labeled such that $\theta_1<\cdots<\theta_n$.  Ex-post payoffs are given by $-(x-\theta)^2$ for the principal and by $-(x-y(\theta))^2$ for the agent. That is, the state represents the principal’s preferred action. On the other hand, the agent’s preferred action is given by the \textit{bias function} $y:\Theta\to\R$. I assume that $y$ is strictly increasing. This provides some minimum degree of alignment in the players’ preferences. The difference between the players’ preferred actions in a given state is denoted by $c(\theta):=y(\theta)-\theta$ and I refer to it as the \textit{conflict of interest} in state $\theta$.

\textbf{Information. }
The agent has private and non-verifiable information about the state of the world. He observes a signal realization $s\in S:=\{s_1,\ldots,s_m\}$, with $m\geq 2$. From the point of view of the principal, states and signals are distributed according to  $f\in\Delta(\Theta\times S)$, which is an $n\times m$ matrix with $ij$-th entry equal to $f(\theta_i,s_j):=Pr_f[ \tilde{\theta}=\theta_i,\tilde{s}=s_j]$. I refer to $f$ as the \textit{true distribution} and assume it has full support, i.e.  $f(\theta,s)>0$ for all $(\theta,s)\in\Theta\times S$. Higher states are more likely after higher signal realizations. In particular, signals are labeled such that $s_j$ indexes the posterior mean $\E_{f}[\tilde{\theta}|s_j]$.

\textbf{Belief design.}
The key feature of the model is that the principal can choose how the agent interprets his private information. Specifically, she chooses a distribution $g\in\Delta(\Theta\times S)$ such that after observing a given signal realization, the agent computes his posterior beliefs according to $g$ (rather than according to $f$).  This step can be interpreted as a selection or hiring process in which the principal chooses an agent who already possesses the desired beliefs.

I impose two restrictions on the set of joint distributions the principal can select.  First, I assume that players \textit{ex ante} agree on the distribution of the state, i.e.,  $g_\Theta=f_\Theta$. Additionally, I require the agent to be well-calibrated about the frequency of signal realizations, i.e., $g_S=f_S$. These restrictions allow us to focus on the information the agent perceives after each signal realization. Let
\[
\mathcal{G}:=\{g\in\Delta(\Theta\times S): g_\Theta=f_\Theta, g_S=f_S\}
\]
denote the set of feasible choices for the principal.\footnote{From a probabilistic perspective, selecting a joint distribution with given marginals corresponds to selecting a copula. I use this terminology only as a mathematical reference; the economic interpretation throughout is in terms of how signals are mapped into posterior beliefs.}

Fixing the marginal distribution of states ensures that the principal and the agent agree \textit{ex ante} on base rates of underlying conditions. Fixing the marginal distribution of signals reflects agreement on the empirical frequency of signal realizations, which can be learned from repeated exposure. Disagreement is therefore confined to how signal realizations are interpreted, that is, to the mapping from signals to posterior beliefs over states. Different elements of $\mathcal{G}$ correspond to different inference rules that merge the same distributions of states and signals.

Throughout the analysis, belief design should be understood as a methodological device rather than as a literal description of how firms intervene in employees’ beliefs. In practice, firms observe a pool of candidates who differ along several dimensions, including preferences, information, and worldviews. The analysis abstracts from heterogeneity in preferences and information quality and focuses instead on variation in how candidates interpret information. Under this interpretation, the solution to the belief design problem identifies the type of agent (identified by his beliefs) that a firm would optimally select for a given position, given the informational and incentive requirements of the role.\footnote{Alternatively, it can also reflect the principal’s ability to  (costlessly) train the agent on how to interpret his information. See \citet{gervais2001learning}, \citet{haran2010simple}, and \citet{meikle2016overconfidence} for practices that can mitigate or exacerbate miscalibration.}

\textbf{Timing.}
The timing of events is as follows:
\begin{enumerate}
	\item Belief design: the principal chooses $g\in\mathcal{G}$.
	\item Nature draws $(\theta,s)$ according to $f$.
	\item The agent observes $s$, interpreting it according to $g$, and chooses $x\in\R$.
\end{enumerate}
Note that the description of the timing assumes that the agent learns nothing from the principal’s choice. This is, interpreting belief design as a selection process, candidates’ beliefs remain the same whether they are selected by the principal or not.

\section{Applications and Definitions}

\label{sec:appsanddefs}

\textbf{Applications.}
I now discuss two specific applications. The purpose is to illustrate the type of situations in which the forces captured by the model are relevant. I later use these applications to illustrate the implications of the results.

First, consider a CEO (the principal) selecting a middle manager (the agent) to oversee the compensation of a subordinate. Specifically, the manager chooses a reward level $x\in\R$ for his subordinate after privately observing a signal of the subordinate’s actual performance. The state represents the CEO’s ideal level of reward given the subordinate’s actual performance, which is never directly observed. The signal is normalized to represent the CEO’s ideal level of reward based on a noisy measure of the subordinate’s performance.

A plausible concern for the CEO is that the manager would be overly reluctant to provide low rewards to his subordinate. In other words, while the two parties may be relatively aligned when the subordinate deserves a bonus or a promotion, the manager may be averse to firing the subordinate, even when the CEO wants exactly that.

Second, consider the director of a nuclear power plant (the principal) selecting a risk manager (the agent) to monitor and control safety in the plant’s operations. The manager chooses a level of risk abatement $x\in\R$ after having privately observed a signal about the safety conditions of the operations. The state represents the director’s ideal abatement level based on the plant's actual performance. As in the previous application, the signal represents the director’s ideal abatement level based on a noisy measure of the plant's performance. In this case, both individuals would agree on the right course of action when the risk is high: avoiding an accident is a shared goal when calamities are likely. On the other hand, when the risk is sufficiently low, the manager does not see the need to keep abatement at levels desired by the director.

In both contexts, I ask: when is an overconfident or an underconfident manager optimal? This question, however, lacks a complete meaning without a proper definition of overconfidence. This is precisely the next step.

\textbf{Defining overconfidence.}
My goal is to provide conditions under which the optimal agent displays beliefs that can be interpreted as  over- or under-precision. If the optimal agent’s beliefs coincide with those prescribed by the true distribution $f$, then I say that he is \textit{well-calibrated}. Otherwise, I need to compare a given solution $g^*$ with $f$. I propose using a well-known stochastic order to make such comparisons that reflect the notion of over- or under-precision. This order is based on the concept of concordance, which roughly corresponds to large values of the state going together with large values of the signal. Thus, an increase in the concordance between these two random variables can be interpreted as the signal ``revealing more’’ about the state.

\begin{definition}[Concordance Order]\footnote{See \cite{tchen1980inequalities}, \cite{epstein1980increasing}, \citet[Chapter~9]{shaked2007stochastic}, \cite{meyer2012increasing}, and \cite{mekonnenbayesian}.}
	\label{def:overconfidence}
	$g\in\mathcal{G}$ dominates $f$ in the \textbf{concordance order}, denoted $g\succeq f$, if and only if for all $k\in \{1,\ldots,n\}$ and $l\in \{1,\ldots,m\}$ we have
	\[
	\sum_{i=1}^{k}\sum_{j=1}^{l}g(\theta_i,s_j)\geq \sum_{i=1}^{k}\sum_{j=1}^{l}f(\theta_i,s_j).
	\]
\end{definition}

The concordance order ranks two distributions by comparing their cumulative functions pointwise. When $g\succeq f$, the probability that the realizations of the state and the signal are both ``small’’ is higher under $g$ than under $f$. Since $g$ and $f$ have the same marginals, it is also true that the probability that the realizations of the state and the signal are both high increases when the distribution changes from $f$ to $g$ \citep[][Theorem~3]{epstein1980increasing}. Therefore, an agent with beliefs $g\succeq f$ interprets higher signal realizations as \textit{stronger} evidence of higher states than his well-calibrated counterpart. To the extent that overprecision is informally understood as excessive faith in one’s information, this definition captures exactly that.\footnote{In the present two-dimensional setting, other orders—the supermodular stochastic order, greater weak association, the convex-modular order, and the dispersion order—coincide with the concordance order \citep[see][]{meyer2012increasing}.}

Additionally, $g\succeq f$ implies that the Pearson correlation coefficient, Kendall rank correlation coefficient, and Spearman's rank correlation coefficient between the state and the signal are all higher under $g$ than under $f$ \citep{tchen1980inequalities}. The converse is not true in general. Thus, the concordance order is more conservative than, for instance, comparing covariances as a criterion to define overprecision, which also lacks a strong justification beyond the multivariate-normal case.

In canonical environments, this definition coincides with more familiar notions of overprecision. For instance, in the multivariate-normal  or truth-or-noise models, an increase in concordance between the state and the signal is equivalent to a reduction in the agent’s subjective posterior variance. In these settings, an agent whose beliefs dominate the true distribution in the concordance order places more weight on the signal and responds more aggressively to information, as in standard models of overprecision. The advantage of the concordance order is that it extends this intuition to nonparametric information structures in which posterior variance alone does not fully characterize belief responses.

Finally, note that it is possible to have a distribution $g\preceq f$ that reverses the relationship between states and signal realizations, such that a high signal becomes evidence of a low state. This extreme change in the interpretation of the signal is not compatible with the idea of overprecision. However, given the maintained assumptions, this issue does not arise in the present setting. In particular, $y(\cdot)$ being increasing guarantees that the solution to the belief design problem is always above, in the concordance order, the independent distribution; that is, any solution $g^*$ satisfies $g^*(\theta,s)\geq f_\Theta(\theta)f_S(s)$ for all $(\theta,s)$.

Thus, if the solution satisfies $g^*\neq f$ and $g^*\succeq f$, I say that the optimal agent is \textit{overconfident}. Analogously, when $f\succeq g^*$, I say that he is \textit{underconfident}.

\section{Analysis}
\label{sec:analysis}
For any $g\in\Delta(\Theta\times S)$ and signal realization $s\in S$, the agent optimally chooses $\E_g[y(\tilde{\theta})|s ]$. Note that all feasible beliefs yield the same average action:
\begin{equation*}
	\E_f[\E_g[y(\tilde{\theta})|s ]]=\E_g[\E_g[y(\tilde{\theta})|s ]]=\E_g[y(\tilde{\theta})]=\E_f[y(\tilde{\theta})]
\end{equation*}
where the first equality follows from $f_S=g_S$, the second from the law of iterated expectations, and the third from $g_\Theta=f_\Theta$.

Knowing the agent’s actions, I can express the principal’s expected payoff as a function of any belief $g$ as  $U(g):=-\E_f[(\E_g[y(\tilde{\theta})|\tilde{s}]-\tilde{\theta})^2]$. As a result, belief design corresponds to the following optimization problem:
\[
\max_{g\in\mathcal{G}} \ U(g).
\]
I will divide the analysis of this problem into two parts. As a first step, I focus on the simplest version of the model, in which the state and signal spaces are binary. Then, I discuss the key insights that extend to the general version of the model. Proofs are provided in \autoref{sec:proofs}.
\subsection{Binary states and signals}
\label{subsec:binary}
Assume that $n=m=2$. For any true distribution $f$, the set of feasible choices for the principal can be characterized by a single scalar. Formally, any  $g\in\mathcal{G}$ can be decomposed as $f+\tau b b^\top$ where $\tau\in\R$ is a scalar and $b:=(1,-1)^\top$ is a fixed vector.\footnote{Thus, the decomposition looks as follows\begin{align*}
		g=\left[
		\begin{tabular}{c c}
			$f(\theta_1,s_1)+\tau$ & $f(\theta_1,s_2)-\tau$ \\
			$f(\theta_2,s_1)-\tau$ & $f(\theta_2,s_2)+\tau$ 
		\end{tabular}
		\right]=\left[
		\begin{tabular}{c c}
			$f(\theta_1,s_1)$ & $f(\theta_1,s_2)$ \\
			$f(\theta_2,s_1)$ & $f(\theta_2,s_2)$ 
		\end{tabular}
		\right]+\tau \left[
		\begin{tabular}{c c}
			$1$ & $-1$ \\
			$-1$ & $1$ 
		\end{tabular}
		\right]=f+\tau b b^\top.
\end{align*}} 
 To see why this is true, suppose we modify the probability of some pair of states and signals, say $(\theta_1,s_1)$, by adding some amount $\tau$. Since marginal probabilities must remain constant, we need to adjust the probability of $(\theta_1,s_2)$ and  $(\theta_2,s_1)$ by adding $-\tau$. As a consequence, we also need to add $\tau$ to the probability of $(\theta_2,s_2)$. The scalar $\tau$ parameterizes what is commonly referred to as an \textit{elementary transformation} \cite[see][]{epstein1980increasing}. 

Furthermore, $\tau$ is constrained by the fact that the resulting $g$ must be a well-defined probability mass function, i.e., $g(\theta_i,s_j)\in[0,1]$. Consequently, $\tau$ must belong to some compact interval $[\underline{\tau},\bar{\tau}]$.\footnote{The bounds on the interval are as follows \begin{align*}
		\underline{\tau}&:=-\min\left\{f(\theta_1,s_1),1-f(\theta_1,s_2),1-f(\theta_2,s_1),f(\theta_2,s_2) \right\}\\
		\bar{\tau}&:=\min\left\{1-f(\theta_1,s_1),f(\theta_1,s_2),f(\theta_2,s_1),1-f(\theta_2,s_2) \right\}.
\end{align*}}
The full support assumption  guarantees that $\underline{\tau}<0<\bar{\tau}$. In other words, any marginal deviation from the true distribution is feasible.

As a result, the belief design problem can be seen as choosing $\tau\in[\underline{\tau},\bar{\tau}]$ to maximize $U( f+\tau b b^\top)$, which is a well-behaved concave maximization program. Before discussing its solution, note that an increase in $\tau$ corresponds to a shift in probability mass from the off-diagonal of $g$ to its diagonal. That is, the agent’s belief that the state is $\theta_i$ after observing $s_i$ increases with $\tau$. Therefore, I interpret $\tau$ as the agent’s \textit{level of confidence}: an agent with $\tau>0$ is \textit{overconfident},  with $\tau<0$ is \textit{underconfident}, and \textit{well-calibrated} otherwise. This is both intuitive and consistent with the definition of overconfidence proposed in \autoref{sec:appsanddefs}.

Increasing $\tau$ amplifies the agent’s responsiveness to signals, pushing actions closer to $y(\theta_i)$ after observing $s_i$. Consequently, a higher $\tau$ raises the action after $s_2$ and lowers it after $s_1$. Overprecision benefits the principal when she would like to reduce the agent’s action more after $s_1$ than she would like to increase it after $s_2$. Formally, this requires $c(\theta_1) > c(\theta_2)$, or equivalently $\theta_2 - \theta_1 > y(\theta_2) - y(\theta_1)$. Intuitively, increasing $\tau$ reallocates responsiveness toward the signal realization associated with the larger misalignment.

The following result describes the optimal agent in the binary case.
\begin{proposition}
	\label{prop:2x2}
	The unique optimal agent is
	\begin{enumerate}
		\item well-calibrated if and only if $\theta_2-\theta_1=  y(\theta_2)-y(\theta_1)$.
		\item overconfident if and only if  $\theta_2-\theta_1>  y(\theta_2)-y(\theta_1)$.
		\item underconfident if and only if $\theta_2-\theta_1< y(\theta_2)-y(\theta_1)$.
	\end{enumerate}
\end{proposition}
Alternatively, the correlation between the conflict of interest and the agent’s preferred action determines the confidence of the optimal agent. In particular, the optimal agent is overconfident when such a correlation is negative.

This result implies that a well-calibrated agent is optimal if and only if his bias is additive, i.e., $y(\theta)=\theta+a$ for some $a\in\R$.  This extends to the optimal agent as well; he would ideally behave \textit{as if} his bias were additive. However, the bounds on $\tau$ may prevent the principal from getting that far.

In addition, the interior optimum level of confidence is given by
\begin{align*}
	\tau^*=-\frac{Cov_f(c(\tilde{\theta}),y(\tilde{\theta}))}{Var_f(y(\tilde{\theta}))}\abs{f}=[\frac{Cov_f(\tilde{\theta},y(\tilde{\theta}))}{Var_f(y(\tilde{\theta}))}-1]\abs{f},
\end{align*}
where $\abs{f}$ denotes the determinant of the matrix $f$.

Since signal realizations are labeled so as to index the induced conditional expectation over states, it follows that $\abs{f}>0$. If $y(\cdot)$ is increasing, this directly implies that $\tau^*>-\abs{f}\in(\underline{\tau},0)$. An agent with a level of confidence $\tau=-\abs{f}$ acts as if the signal were uninformative; therefore, the optimal agent always places some positive weight on the signal.

Moreover, the optimal level of confidence increases with the slope of the best affine predictor of the principal’s preferred action, $\tilde{\theta}$, based on the agent’s preferred action, $y(\tilde{\theta})$. If, for given primitives, $\tau^*>\bar{\tau}$, then the optimal agent is maximally overconfident.

I now illustrate the implications of the result for the applications described in \autoref{sec:appsanddefs}. In the CEO-manager example, the CEO's concern about the manager’s reluctance to punish a subordinate is captured by assuming $y(\theta_1)>\theta_1$ and $y(\theta_2)\approx\theta_2$. It follows that $c(\theta_1)>c(\theta_2)$, and by  \autoref{prop:2x2}, the CEO would strictly prefer to hire an \textit{overconfident} manager.

By contrast, in the nuclear power plant example, the friction stems from the manager’s reluctance to mitigate risks when an accident is unlikely. This can be represented by assuming $y(\theta_1)<\theta_1$ and $y(\theta_2)\approx\theta_2$. In this case, $c(\theta_1)<c(\theta_2)$, and the optimal risk manager is \textit{underconfident}.

\subsection{General case}
\label{subsec:rational}
In this subsection, I consider the general case with an arbitrary number of states and signals. As in the binary model, the principal’s choice is a joint distribution with fixed marginals, which determines how the agent maps signal realizations into posterior beliefs. Because the agent’s action depends only on these posterior beliefs, the principal’s expected payoff depends on the choice of $g$ only through the agent’s conditional expectations. This observation allows the analysis of the general case to proceed along the same lines as in the binary setting.

I begin by noting that the principal’s objective in the belief design problem admits the following decomposition:
\begin{align*}
	U(g)=-\E_f[(\E_f[\tilde{\theta}|\tilde{s}]-\tilde{\theta})^2 ]-\E_f[c(\tilde{\theta})]^2-Var_f(\E_g[y(\tilde{\theta})|\tilde{s}] -\E_f[\tilde{\theta}|\tilde{s}]).
\end{align*}
The first term represents the payoff that the principal could obtain if she were informed and in charge of choosing the action. It corresponds to a loss due to the residual uncertainty in the environment. The second term reflects a loss due to the average bias introduced by the agent's choice. Finally, the agent’s beliefs affect the objective only through the last term, which represents a loss due to variance in the agent’s actions beyond the adjustments that the principal herself would make.

This decomposition illustrates that the principal ideally wants $\E_g[y(\tilde{\theta})|\tilde{s}] -\E_f[\tilde{\theta}|\tilde{s}]$ to be constant.  In other words, the ideal agent is one who acts \textit{as if} his bias were additive. Therefore, if the expected conflict of interest is invariant in the signal, i.e. $\E_f[y(\tilde{\theta})|\tilde{s}] =\E_f[\tilde{\theta}|\tilde{s}]+b$, the well-calibrated agent would be optimal since he is feasible and already behaves as the principal ideally wants. It turns out that this condition is also necessary. Whenever the conditional expectation of the conflict of interest varies with the signal, there is a marginal deviation from the true distribution that strictly increases the principal’s expected payoff. This is achieved by decreasing the agent’s action after a signal realization that leads to a higher expected conflict of interest, while increasing the agent’s action after a realization that leads to a lower expected conflict of interest. The construction of a distribution that improves upon the true one follows the same logic as in the binary case; the details are discussed in the proof of the following result (see \autoref{sec:proofs}).

\begin{proposition}
	\label{prop:nxm1}
	The optimal agent is well-calibrated if and only if  \ $\E_{f}[y(\tilde{\theta})-\tilde{\theta}|\tilde{s}]$ is constant.
\end{proposition}

Moreover, because the marginals of $g$ equal those of $f$, $\E_g[y(\tilde{\theta})|\tilde{s}] -\E_f[\tilde{\theta}|\tilde{s}]$ can only equal one constant,  which is $\E_f[c(\tilde{\theta})]$. Therefore, for this ideal agent, we have that
$\E_g[y(\tilde{\theta})|\tilde{s}] -\E_f[y(\tilde{\theta})|\tilde{s}]=\E_f[c(\tilde{\theta})]-\E_f[c(\tilde{\theta})|\tilde{s}]$. When $\E_f[c(\tilde{\theta})|s]$ is decreasing in $s$, the ideal agent’s optimal action would be below the well-calibrated agent’s action for \textit{low} signals realizations, while the opposite is true for \textit{high} ones. Since $\E_f[y(\tilde{\theta})|s]$ is increasing, this pattern corresponds to more extreme actions by the agent, a direct manifestation of overconfidence. As a result, the idea that overconfidence is optimal when the principal would adjust the action more than the agent generalizes beyond the binary special case. The following result formalizes this intuition.
\begin{proposition}
	\label{prop:nxm2}
	There exists $\bar{\alpha}>0$ such that if $\abs{  \E_{f}[c(\tilde{\theta})|s]-\E_{f}[c(\tilde{\theta})]  }\leq \bar{\alpha}$
	for all $s\in S$, then
	\begin{itemize}
		\item $\E_{f}[y(\tilde{\theta})|s’]-\E_{f}[y(\tilde{\theta})|s]<\E_{f}[\tilde{\theta}|s’]-\E_{f}[\tilde{\theta}|s]$ for all $s’>s$ implies that any optimal agent acts as an overconfident agent.
		
		\item $E_{f}[y(\tilde{\theta})|s']-\E_{f}[y(\tilde{\theta})|s]>\E_{f}[\tilde{\theta}|s']-\E_{f}[\tilde{\theta}|s]$ for all $s'>s$ implies that any optimal agent acts as an underconfident agent. 
	\end{itemize}
\end{proposition}

The previous result provides sufficient conditions for at least one optimal agent to be strictly ranked above or below the true distribution according to the concordance order.   The uniqueness result from the binary case is necessarily lost since there exist multiple distributions with the same marginals and conditional expectations. However, all optimal agents induce the same actions. Moreover, all optimal agents must be unranked among themselves; that is, if $g$ and $g’$ solve the belief design problem, it cannot be that $g\succeq g’$.

The results can be interpreted as explaining how different organizational roles generate demand for different employee worldviews. In the model, positions differ along two reduced-form dimensions. First, they differ in how much the firm’s desired action varies with information, captured by the dispersion of $\E_f[\tilde{\theta}|s]$ across signal realizations. Roles operating in volatile environments or at the interface with external markets naturally exhibit greater variation in this object, whereas more protocol-oriented or internally focused roles exhibit less variation. Second, positions differ in how preference misalignment varies across states, captured by the signal-dependent conflict of interest $\E_f[c(\tilde{\theta})|s]$. Together, these primitives determine whether a firm prefers an employee whose interpretation of information is well-calibrated, more extreme, or more muted. Under this interpretation, belief distortions arise endogenously as a response to the informational and incentive requirements of a given role.

The proof of \autoref{prop:nxm2} parallels that of $\autoref{prop:2x2}$. I start by changing the principal’s choice from $g\in \mathcal{G}$ to a matrix of \textit{elementary transformations} $t\in \R^{(n-1)\times(m-1)}$. Let $t_{kl}$ denote a typical entry of the matrix $t$. Note that $t_{kl}>0$ moves probability mass from $\theta_{k+1}$ to $\theta_k$ after  $s_l$ is realized, while the opposite happens after $s_{l+1}$. Informally, a positive elementary transformation moves mass from discordant pairs of states and signal realizations to the adjacent concordant pairs.  The next step is to analyze the first-order conditions that optimal elementary transformations need to satisfy. Then, I propose a family of transformations that satisfy these conditions. When all entries of $t$ are positive, the resulting distribution dominates the initial one in the concordance order. The assumptions in \autoref{prop:nxm2} guarantee that the proposed solution is both feasible and positive (or negative), which proves the existence of an overconfident (or underconfident) optimal agent.

In what follows, I explore the implications of additional tools available to the principal.  I go back to the assumption that $n=m=2$.  First, I will study how belief design is affected by the availability of action-contingent transfers. I will argue that transfers do not affect optimal beliefs as long as the expected conflict of interest is not too far away from zero. Moreover, in the optimum, each tool (belief and contract design) is used for different purposes. Additionally, I consider the possibility that the principal can make the choice herself. This imposes some restrictions on the characteristics of an agent who is actually allowed to make the choice: for the principal to delegate the choice, the agent’s confidence must be sufficiently high.

\section{Transfers}
\label{sec:transfers}
Conflicting preferences among members are a prominent challenge for organizations. The provision of monetary incentives is a particularly relevant tool for mitigating the pernicious effects of agency frictions. In this section, I consider the interaction between belief design and action-contingent transfers. In particular, in addition to belief design, the principal is also allowed to commit to non-negative payments contingent on the agent’s action. I turn the focus back to the $n=m=2$ case. The timing is as follows:
\begin{enumerate}
	\item Belief and contract design: the principal chooses $g\in\mathcal{G}$ and $w:\R\to\R_+$.
	\item Nature draws $(\theta,s)$ according to $f$.
	\item The agent observes $s$ and chooses $x\in\R$.
\end{enumerate}
Payoffs are given by $-(x-\theta)^2-w(x)$ for the principal and by $-(x-y(\theta))^2+w(x)$ for the agent.

As a first step, I set the problem as the principal recommending action $x_i$ after signal realization $s_i$, paying $w_i$ upon observing that action, and paying zero upon observing a non-recommended action. Additionally, in the binary case, belief design can be thought of as choosing a level of confidence $\tau\in[\underline{\tau},\bar{\tau}]$. Therefore, it suffices for the principal to consider tuples $(x_1,x_2,w_1,w_2,\tau)$ consisting of recommended actions, payments for those actions, and a level of confidence subject to obedience constraints.

The recommended actions must be incentive compatible given promised transfers and the agent’s confidence. After each signal realization, two deviations are relevant: to the other recommended action and to the best action among those that were not recommended. The best deviation to an action yielding no transfer corresponds to choosing $\E_\tau[y(\tilde{\theta})|s_i]$ after signal $s_i$.

The main result in this section provides conditions under which the optimal beliefs coincide with those described in \autoref{subsec:binary}, where transfers were not available. In the optimum, each tool plays a different role: transfers reduce the average bias in the agent’s actions, while beliefs distribute that bias across signal realizations.
\begin{proposition}
	\label{prop:transfers}
	Assume that $\abs{\E_{f}[c(\tilde{\theta})]}\leq \E_f[\tilde{\theta}|s_2]-\E_f[\tilde{\theta}|s_1]$. The availability of transfers does not change the optimal beliefs (so that \autoref{prop:2x2} applies) and
	\[w_1^*=w_2^*=\frac{1}{4}\E_f[c(\tilde{\theta})]^2.   \]
\end{proposition}
This result also highlights the effects of belief design on optimal transfers. When beliefs are optimally chosen, wages do not change with the agent's actions in equilibrium. On the other hand, when only the well-calibrated agent is available, the optimum for the principal is given by
\begin{align*}
	x_j=\frac{1}{2}[\E_f[\tilde{\theta}|s_j]+\E_f[y(\tilde{\theta})|s_j]],\quad
	w_j=\frac{1}{4}\E_f[c(\tilde{\theta})|s_j]^2.
\end{align*}
Therefore, unless $\E_f[c(\tilde{\theta})|s_j]$ is constant (in which case the well-calibrated agent is indeed optimal), wages are different for both recommended actions. This is, optimal belief-based selection leads to ``flatter" compensation schemes.

\autoref{prop:transfers} is established under the assumption that both the state and the signal spaces are binary. This restriction is imposed for tractability rather than for conceptual reasons. The economic mechanism underlying the result (the separation between belief design, which governs how the agent’s actions vary with information, and transfers, which affect only the average bias) does not depend on the cardinality of the state or signal spaces.

Extending \autoref{prop:transfers} to the more general information structures considered in \autoref{prop:nxm2} would require a substantially more involved characterization of incentive compatibility, as transfers would need to be defined over a larger set of equilibrium actions and deviations. While such an extension is technically feasible, it is not pursued here. Nevertheless, the mechanisms identified in the binary model point to a limited role for transfers: under comparable conditions on the expected conflict of interest, optimal transfers would plausibly be flat across equilibrium actions and would not interfere with optimal belief design. \autoref{sec:truth} provides support for this intuition by establishing the result in a tractable continuous-state setting with a truth-or-noise information structure. For this reason, the binary case is used as the baseline environment for the analysis of transfers.

It is worth noting that \citet{ashworth2019delegation} studies the interaction between optimal delegation sets and transfers in a similar setting. When the principal and the agent differ only in their levels of confidence, the optimal mechanism does not involve transfers.

\section{Delegation}
\label{sec:delegation}
Consider the case in which the principal decides whether to delegate the decision or centralize it. Under centralization, the principal chooses an action based only on the information available \textit{a priori}. In this section, I characterize the conditions under which the principal prefers to delegate the decision to the agent, and the implications of this choice for the optimal agent’s level of confidence.

The main takeaway is that centralization is a substitute for \textit{extreme} underconfidence. The intuition is simple: if the optimal agent is not \textit{sufficiently} using the information, it would be better for the principal to make the choice herself in order to avoid the conflict of interest. The principal would rely solely on the agent's private information; otherwise, she would be better off avoiding the bias the agent introduces into the decision.

The following result shows that the optimal agent must be sufficiently confident for the decision to be delegated to him.

\begin{proposition}
	\label{prop:delegation}
	The choice is delegated to the optimal agent if and only if
	\[
	\tau^*\geq \frac{f_S(s_1)f_S(s_2)\E_f[c(\tilde{\theta})]^2}{\abs{f}(y(\theta_2)-y(\theta_1))(\theta_2-\theta_1)}-\abs{f}.
	\]
\end{proposition}
Naturally, an agent that is on average unbiased (i.e., $\E_f[c(\tilde{\theta})]=0$) would act just as the principal in the absence of any information. Therefore, when $\E_f[c(\tilde{\theta})]=0$ the choice would be delegated even to a maximally underconfident agent (one who thinks that signals are uninformative). On the other hand, whenever $\E_f[c(\tilde{\theta})]\neq 0$, the agent’s confidence must be strictly above  $-\abs{f}$ (the confidence of a maximally underconfident agent) for delegation to be optimal.

While extreme underconfidence never leads to delegation to an agent that is on average biased, it can be optimal to delegate to an extremely overconfident agent.\footnote{Consider the following example: $\theta_1=0$, $\theta_2=10$, $y(\theta)=3+\theta/3$, $f(\theta_i,s_i)=0.4$, and $f(\theta_i,s_{-i})=0.1$. We have that  $\tau^*=0.3>\bar{\tau}=0.1$, which implies that the optimal agent is maximally overconfident, interpreting signals at face value (after seeing $s_i$ this agent would be convinced that the state is $\theta_i$). The expected payoff from delegation is -17.89, while that from centralization is -25. Therefore, it is optimal to delegate to the maximally overconfident agent.} Therefore, while the two phenomena are modeled symmetrically, the possibility of centralizing decision-making highlights a key practical distinction between overconfidence and underconfidence.

Although \autoref{prop:delegation} is stated for the binary model, the mechanism is not specific to that setting. \autoref{sec:truth} shows that under a truth-or-noise information structure, delegation is optimal if and only if the optimally chosen agent places sufficient weight on the signal, so that extreme underconfidence is again incompatible with delegation.

\section{Conclusions}
\label{sec:conclusion}

There is abundant evidence about the pervasiveness of overconfident employees. This paper examines how organizations optimally select employees with different levels of confidence for decision-making roles when information interacts with conflicts of interest. In some cases, a firm is willing to select an employee not \textit{despite} his overconfidence but precisely \textit{because} of it. Thus, a reason overconfidence is ubiquitous is that the conditions that lead to its optimality are common across many environments. Under-responsiveness to the firm’s interests is “bread-and-butter behavior” when employees prioritize a quiet life over adjusting their behavior to current conditions, and it can be (at least partially) alleviated by the employee’s misperception of being more informed.

Interpreted through this lens, the analysis suggests that employees in roles requiring frequent adjustment to external conditions (such as sales, recruiting, marketing, brokerage, or public relations) are likely to have higher confidence than those working in roles emphasizing stability and adherence to established procedures or protocols (such as auditing, risk management, compliance, IT support, or data processing).

Moreover, belief-based selection interacts meaningfully with other measures that the firm could take to mitigate agency frictions. For example, an employee with optimal beliefs would face ``flatter" compensation schemes than his well-calibrated counterpart. Additionally, centralizing decision-making is a natural substitute for extreme underconfidence, while it can be optimal to delegate to a ``fully" overconfident employee.

Job positions within organizations are inherently multidimensional objects, combining informational requirements, incentive concerns, and institutional constraints. The analysis abstracts from many of these features and instead isolates two first-order characteristics: how information affects the firm’s desired action and how preference misalignment varies across states. This parsimonious structure allows for a tractable characterization of how different roles generate demand for different employee worldviews. While richer descriptions of job design may incorporate additional dimensions, the framework highlights a fundamental channel through which organizational roles shape the level of confidence that firms optimally seek in their employees.

The main methodological innovation, the belief design problem, allows the firm to exploit the heterogeneity in applicants’ confidence by selecting the candidate with the most favorable beliefs (keeping fixed all other individual characteristics). In this context, I show that belief-based selection leads to employees with a common feature: they tend to act \textit{as if} their disagreement with the firm was invariant to their private information. As a result, the firm prefers a well-calibrated agent if and only if this disagreement does not change with the employee’s observations to begin with. On the other hand, overconfidence helps when this conflict of interest runs counter to the employee’s preferred action. As a tool, belief design is flexible and can be adjusted to systematically examine other belief-based biases, such as overoptimism and wishful thinking.

Two assumptions were maintained throughout this exercise: that the potential employees are sufficiently diverse in their beliefs and that the firm perfectly observes their characteristics. The first assumption is made to provide a benchmark. The second is motivated by evidence suggesting that firms can learn quite rapidly about their employees’ characteristics \citep{lange2007speed,gradesandlearning}. Issues related to a constrained set of available employees and asymmetric information about applicants’ characteristics are promising avenues for future research.

These findings illustrate interesting ways in which personality traits, such as confidence, can affect labor market outcomes, including wages and career choices \citep[see][]{schulz2016overconfidence}. Similarly, they help explain several documented managerial and organizational practices targeted towards altering employees’ perceptions  \citep[see][]{haran2010simple,meikle2016overconfidence}. Additionally, if confidence has the potential to affect expected outcomes (in the labor market or otherwise), we may expect individuals to invest in ``adjusting" this personal characteristic according to their goals.\footnote{See \citet{KrepsNemmers} for a discussion about the role of business schools in boosting students’ confidence.}

More broadly, the analysis highlights a fundamental asymmetry between inducing action and restraining it in delegated decision-making. In many relationships, it may be easier to discipline excessive responses than to motivate agents who are intrinsically reluctant to respond to imperfect information. From this perspective, overconfidence is not valuable because it improves judgment, but because it lowers the threshold for action under uncertainty. Similar considerations arise in other contexts, such as parenting or supervision, where setting boundaries on overly proactive behavior may be less costly than pushing disengaged individuals to act. Decisions must often be taken before outcomes can be assessed. Organizations may therefore favor agents who are willing to act decisively, even at the risk of error. This view suggests a broader role for belief distortions in facilitating action when waiting for validation is itself costly. An interesting direction for future research is to model the relative costs of restraining versus motivating behavior explicitly, and to study how these asymmetries shape the selection of agents across a broader range of organizational and social relationships.
\newpage

\section*{References}
\bibliography{confidence}

\newpage

\appendix
	
	\section{Proofs}
	\label{sec:proofs}
	\renewcommand{\thedefinition}{\Alph{section}.\arabic{definition}}
	\renewcommand{\theassumption}{\Alph{section}.\arabic{assumption}}
	\renewcommand{\theproposition}{\Alph{section}.\arabic{proposition}}
	\renewcommand{\thelemma}{\Alph{section}.\arabic{lemma}}
	\renewcommand{\theequation}{\Alph{section}.\arabic{equation}}
	
	\setcounter{equation}{0}
	\setcounter{definition}{0}
	\setcounter{assumption}{0}
	\setcounter{proposition}{0}
	\setcounter{lemma}{0}
	
	\begin{proof}[Proof of \autoref{prop:2x2}]
		The belief design problem is
		\[
		\max_{\tau\in[\underline{\tau},\bar{\tau}]} U\!\left(f+\tau b b^\top\right).
		\]
		For $\tau\in[\underline{\tau},\bar{\tau}]$, the induced posteriors satisfy
		\begin{align*}
			\E_\tau[y(\tilde{\theta})\mid s_1]
			&=\E_f[y(\tilde{\theta})\mid s_1]-\frac{\tau}{f_S(s_1)}\bigl(y(\theta_2)-y(\theta_1)\bigr),\\
			\E_\tau[y(\tilde{\theta})\mid s_2]
			&=\E_f[y(\tilde{\theta})\mid s_2]+\frac{\tau}{f_S(s_2)}\bigl(y(\theta_2)-y(\theta_1)\bigr).
		\end{align*}
		Hence,
		\begin{align*}
			U(f+\tau b b^\top)
			=&-\sum_{i=1}^2\sum_{j=1}^2 f(\theta_i,s_j)\bigl(\E_\tau[y(\tilde{\theta})\mid s_j]-\theta_i\bigr)^2 .
		\end{align*}
		Differentiating and simplifying yields
		\begin{align*}
			\frac{\partial U(f+\tau b b^\top)}{\partial \tau}
			=
			2\bigl(y(\theta_2)-y(\theta_1)\bigr)\Bigl[
			\bigl(\E_\tau[y(\tilde{\theta})\mid s_1]-\E_f[\tilde{\theta}\mid s_1]\bigr)
			-
			\bigl(\E_\tau[y(\tilde{\theta})\mid s_2]-\E_f[\tilde{\theta}\mid s_2]\bigr)
			\Bigr].
		\end{align*}
		Moreover,
		\[
		\frac{\partial^2 U(f+\tau b b^\top)}{\partial \tau^2}
		=
		-2\frac{\bigl(y(\theta_2)-y(\theta_1)\bigr)^2}{f_S(s_1)f_S(s_2)}<0,
		\]
		so $U(f+\tau b b^\top)$ is strictly concave in $\tau$. The (unique) interior optimizer therefore solves the first-order condition
		\[
		\E_\tau[y(\tilde{\theta})\mid s_1]-\E_f[\tilde{\theta}\mid s_1]
		=
		\E_\tau[y(\tilde{\theta})\mid s_2]-\E_f[\tilde{\theta}\mid s_2].
		\]
		Using $c(\theta)=y(\theta)-\theta$, this is equivalent to
		\[
		\E_f[c(\tilde{\theta})\mid s_1]-\frac{\tau}{f_S(s_1)}\bigl(y(\theta_2)-y(\theta_1)\bigr)
		=
		\E_f[c(\tilde{\theta})\mid s_2]+\frac{\tau}{f_S(s_2)}\bigl(y(\theta_2)-y(\theta_1)\bigr).
		\]
		Next, note that
		\[
		\E_f[c(\tilde{\theta})\mid s_2]-\E_f[c(\tilde{\theta})\mid s_1]
		=
		\bigl(c(\theta_2)-c(\theta_1)\bigr)\,
		\frac{f(\theta_2,s_2)f(\theta_1,s_1)-f(\theta_1,s_2)f(\theta_2,s_1)}{f_S(s_1)f_S(s_2)}.
		\]
		Solving the first-order condition for $\tau$ gives the interior candidate
		\[
		\tau^*
		=
		-\frac{c(\theta_2)-c(\theta_1)}{y(\theta_2)-y(\theta_1)}
		\Bigl[f(\theta_2,s_2)f(\theta_1,s_1)-f(\theta_1,s_2)f(\theta_2,s_1)\Bigr].
		\]
		Equivalently, letting $\abs{f}:=f(\theta_1,s_1)f(\theta_2,s_2)-f(\theta_1,s_2)f(\theta_2,s_1)$ be the determinant of $f$, I can write
		\[
		\tau^*=
		-\frac{c(\theta_2)-c(\theta_1)}{y(\theta_2)-y(\theta_1)}\,\abs{f}.
		\]
		In the binary case, the assumption that higher states are more likely after higher signal realizations is equivalent to
		\[
		f(\theta_2,s_2)f(\theta_1,s_1)>f(\theta_1,s_2)f(\theta_2,s_1).
		\]
		Hence, the sign of $\tau^*$ is determined by the sign of
		\[
		-\frac{c(\theta_2)-c(\theta_1)}{y(\theta_2)-y(\theta_1)}
		=
		-\frac{(\theta_2-\theta_1)-(y(\theta_2)-y(\theta_1))}{y(\theta_2)-y(\theta_1)}.
		\]
		It follows that
		\begin{itemize}
			\item if $\theta_2-\theta_1=y(\theta_2)-y(\theta_1)$, then $\tau^*=0$ and the optimal agent is well-calibrated;
			\item if $\theta_2-\theta_1>y(\theta_2)-y(\theta_1)$, then $\tau^*>0$ and the optimal agent is overconfident;
			\item if $\theta_2-\theta_1<y(\theta_2)-y(\theta_1)$, then $\tau^*<0$ and the optimal agent is underconfident.
		\end{itemize}
		If $\tau^*\notin[\underline{\tau},\bar{\tau}]$, strict concavity implies that the optimum lies at a boundary point, which preserves the same sign and therefore the same classification.
		
	\end{proof}
	
	\begin{proof}[Proof of \autoref{prop:nxm1}]
		If $\E_f[c(\tilde{\theta})\mid \tilde{s}]$ is constant, then the well-calibrated agent is optimal by the payoff decomposition in \autoref{subsec:rational}.
		
		Conversely, suppose $\E_f[c(\tilde{\theta})\mid \tilde{s}]$ is not constant. Then there exist signals $s_l$ and $s_{l'}$ such that
		\[
		\E_f[c(\tilde{\theta})\mid s_{l'}]>\E_f[c(\tilde{\theta})\mid s_l].
		\]
		Fix any two states $\theta_k<\theta_{k'}$. For $\tau\in\mathbb{R}$, define $g_\tau\in\Delta(\Theta\times S)$ by
		\[
		g_\tau(\theta,s)
		=
		\begin{cases}
			f(\theta,s)+\tau,
			& \text{if } (\theta,s)\in\{(\theta_{k'},s_{l'}),(\theta_k,s_l)\},\\[4pt]
			f(\theta,s)-\tau,
			& \text{if } (\theta,s)\in\{(\theta_k,s_{l'}),(\theta_{k'},s_l)\},\\[4pt]
			f(\theta,s),
			& \text{otherwise.}
		\end{cases}
		\]
		
		This perturbation preserves both marginals: for every $\theta$,
		$\sum_{s} g_\tau(\theta,s)=\sum_{s} f(\theta,s)$, and for every $s$,
		$\sum_{\theta} g_\tau(\theta,s)=\sum_{\theta} f(\theta,s)$.
		Hence $g_\tau\in\mathcal{G}$ whenever it is nonnegative.
		
		Let
		\begin{align*}
			\underline{\tau}
			&:=
			-\min\Bigl\{
			f(\theta_{k'},s_{l'}),\ f(\theta_k,s_l),\ 1-f(\theta_k,s_{l'}),\ 1-f(\theta_{k'},s_l)
			\Bigr\},\\
			\bar{\tau}
			&:=
			\min\Bigl\{
			1-f(\theta_{k'},s_{l'}),\ 1-f(\theta_k,s_l),\ f(\theta_k,s_{l'}),\ f(\theta_{k'},s_l)
			\Bigr\}.
		\end{align*}
		Then $g_\tau\ge 0$ for all $\tau\in[\underline{\tau},\bar{\tau}]$.
		By full support, $\underline{\tau}<0<\bar{\tau}$, so this interval contains a neighborhood of $0$.
		
		Under $g_\tau$, posterior beliefs are unchanged for all signals except $s_l$ and $s_{l'}$, and the induced action satisfies
		\[
		\E_{g_\tau}[y(\tilde{\theta})\mid s]
		=
		\begin{cases}
			\E_f[y(\tilde{\theta})\mid s]-\tau\,\dfrac{y(\theta_{k'})-y(\theta_k)}{f_S(s_l)},
			& \text{if } s=s_l,\\[10pt]
			\E_f[y(\tilde{\theta})\mid s]+\tau\,\dfrac{y(\theta_{k'})-y(\theta_k)}{f_S(s_{l'})},
			& \text{if } s=s_{l'},\\[10pt]
			\E_f[y(\tilde{\theta})\mid s],
			& \text{otherwise.}
		\end{cases}
		\]
		A direct substitution into the quadratic-loss objective yields
		\begin{align*}
			U(g_\tau)
			=
			U(f)
			&-\tau^2\bigl(y(\theta_{k'})-y(\theta_k)\bigr)^2
			\left(\frac{1}{f_S(s_{l'})}+\frac{1}{f_S(s_l)}\right)\\
			&-2\tau\bigl(y(\theta_{k'})-y(\theta_k)\bigr)
			\Bigl(\E_f[c(\tilde{\theta})\mid s_{l'}]-\E_f[c(\tilde{\theta})\mid s_l]\Bigr).
		\end{align*}
		Therefore,
		\[
		\frac{\partial U(g_\tau)}{\partial \tau}\Big|_{\tau=0}
		=
		-2\bigl(y(\theta_{k'})-y(\theta_k)\bigr)
		\Bigl(\E_f[c(\tilde{\theta})\mid s_{l'}]-\E_f[c(\tilde{\theta})\mid s_l]\Bigr).
		\]
		Since $y$ is strictly increasing, $y(\theta_{k'})-y(\theta_k)>0$, and by assumption
		$\E_f[c(\tilde{\theta})\mid s_{l'}]-\E_f[c(\tilde{\theta})\mid s_l]>0$.
		Hence the derivative at $0$ is nonzero, so for $\tau$ of the appropriate sign and sufficiently small (thus feasible) $U(g_\tau)>U(f)$ holds.
		It follows that $f$ (the well-calibrated agent) cannot be optimal whenever $\E_f[c(\tilde{\theta})\mid \tilde{s}]$ is not constant.
	\end{proof}

	\begin{proof}[Proof of \autoref{prop:nxm2}]
	I begin with the observation that any joint distribution $g\in\mathcal{G}$ can be reexpressed using a matrix $t$ of \textit{elementary transformations} and two fixed canonical operators. Formally, let $D_k\in\R^{(k-1)\times k}$ denote the canonical \textit{discrete first-derivative operator}.\footnote{Formally, if $I_k$ denotes the $k\times k$ identity matrix, define
		\[
		D_k:=\begin{bmatrix} I_{k-1} & 0 \end{bmatrix}-\begin{bmatrix} 0 & I_{k-1} \end{bmatrix}.
		\]
		For any vector $x=(x_1,\ldots,x_k)^\top\in\R^k$, we have
		\[
		D_kx=(x_1-x_2,\;x_2-x_3,\;\ldots,\;x_{k-1}-x_k)^\top.
		\]
		Thus, $D_k$ computes consecutive first differences. It coincides with the forward finite-difference matrix in numerical analysis, the first-differencing operator in time-series econometrics, and the incidence (discrete gradient) matrix of a path graph in graph theory.}
	Let $t\in\R^{(n-1)\times(m-1)}$ be a matrix whose typical entry $t_{ij}$ corresponds to a local reallocation of probability mass between adjacent states $(\theta_i,\theta_{i+1})$ and adjacent signals $(s_j,s_{j+1})$. Then any $g\in\mathcal{G}$ can be written as
	\[
	g=f+D_n^\top t D_m,
	\]
	for some (not necessarily unique) matrix $t\in\R^{(n-1)\times(m-1)}$.
	
	Note that each $t_{kl}$ affects posterior beliefs only after signal realizations $s_l$ and $s_{l+1}$. Moreover, $t_{kl}>0$ moves probability mass from $\theta_{k+1}$ to $\theta_k$ after $s_l$, while the opposite happens after $s_{l+1}$.
	
	Let $\delta(t,s_j)$ denote the difference between the agent's action after signal realization $s_j$ when beliefs are induced by $t$ and the action of the well-calibrated agent:
	\begin{align*}
		\delta(t,s_j)
		:=&\E_g[y(\tilde{\theta})\mid s_j]-\E_f[y(\tilde{\theta})\mid \tilde{s}=s_j]=\frac{1}{f_S(s_j)}\sum_{i=1}^{n-1}\bigl[y(\theta_{i+1})-y(\theta_i)\bigr]\bigl[t_{i,j-1}-t_{ij}\bigr],
	\end{align*}
	with the convention $t_{0j}=t_{i0}=t_{nj}=t_{im}=0$. In particular, since $t_{i0}=t_{im}=0$,
	\begin{align*}
		\delta(t,s_1)&=-\frac{1}{f_S(s_1)}\sum_{i=1}^{n-1}\bigl[y(\theta_{i+1})-y(\theta_i)\bigr]t_{i1}, \ \text{and} \ 
		\delta(t,s_m)&=\frac{1}{f_S(s_m)}\sum_{i=1}^{n-1}\bigl[y(\theta_{i+1})-y(\theta_i)\bigr]t_{i,m-1}.
	\end{align*}
	Consequently,
	\begin{align*}
		\frac{\partial \delta(t,s_j)}{\partial t_{kl}}
		=
		\left\{
		\begin{array}{lll}
			-\dfrac{y(\theta_{k+1})-y(\theta_k)}{f_S(s_l)}, & \text{if} & j=l,\\[8pt]
			\phantom{-}\dfrac{y(\theta_{k+1})-y(\theta_k)}{f_S(s_{l+1})}, & \text{if} & j=l+1,\\[8pt]
			0, & & \text{otherwise.}
		\end{array}
		\right.
	\end{align*}
	
	On the other hand, $g\in\mathcal{G}$ if and only if
	\[
	t_{i-1,j-1}-t_{i-1,j}-t_{i,j-1}+t_{ij}\in\bigl[-f(\theta_i,s_j),\,1-f(\theta_i,s_j)\bigr]
	\]
	for all $i\in\{1,\ldots,n\}$ and $j\in\{1,\ldots,m\}$, where again $t_{0j}=t_{i0}=t_{nj}=t_{im}=0$. Let $\mathcal{T}$ denote the set of feasible transformations. The full-support assumption guarantees that the interior of $\mathcal{T}$, denoted $\mathcal{T}^\circ$, is nonempty (in particular, the $(n-1)\times(m-1)$ zero matrix belongs to $\mathcal{T}^\circ$).
	
	Therefore, the belief design problem corresponds to choosing $t\in\mathcal{T}$ to maximize $U\!\left(f+D_n^\top t D_m\right)$. Since $U\!\left(f+D_n^\top t D_m\right)$ is continuous in $t$ and $\mathcal{T}$ is compact, a solution exists.
	
	Note that $g\succeq f$ if and only if $t\geq 0$, i.e., $t_{ij}\geq 0$ for all $i\in\{1,\ldots,n-1\}$ and $j\in\{1,\ldots,m-1\}$.\footnote{For a more general version and proof of this statement, see \cite{tchen1980inequalities} or \cite{epstein1980increasing}.} Hence, if some optimal $t^*\gneqq 0$ satisfies $t^*\ge 0$, I say that the optimal agent is \textit{overconfident}.
	
	I decompose the objective function as
	\[
	U(g)
	=
	-\E_f\!\left[\delta(t,\tilde{s})^2\right]
	-2\E_f\!\left[\delta(t,\tilde{s})\,\E_f\!\left[c(\tilde{\theta})\mid \tilde{s}\right]\right]
	+U(f).
	\]
	In particular, the objective depends on $t$ only through the induced function $\delta(t,\cdot)$. Since $\delta(t,\cdot)$ is linear in $t$, the set of feasible $\delta$'s is convex (because $\mathcal{T}$ is convex), and the objective is a strictly concave quadratic function of $\delta(t,\cdot)$, the maximization problem admits a unique optimizer in terms of $\delta$. Consequently, all optimal transformations induce the same function $\delta(\cdot)$ and hence the same behavior.
	
	Differentiating yields
	\begin{align*}
		\frac{\partial U(f+D_n^\top t D_m)}{\partial t_{kl}}
		=
		2\bigl[y(\theta_{k+1})-y(\theta_k)\bigr]
		\Bigl[\delta(t,s_l)+\E_f[c(\tilde{\theta})\mid s_l]-\delta(t,s_{l+1})-\E_f[c(\tilde{\theta})\mid s_{l+1}]\Bigr].
	\end{align*}
	Since $y(\theta_{k+1})>y(\theta_k)$, any interior optimum $t^*\in\mathcal{T}^\circ$ must satisfy, for each $l\in\{1,\ldots,m-1\}$,
	\begin{align}\label{eq:FOCdelta}
		\delta(t^*,s_l)+\E_f[c(\tilde{\theta})\mid s_l]
		=
		\delta(t^*,s_{l+1})+\E_f[c(\tilde{\theta})\mid s_{l+1}].
	\end{align}
	\autoref{eq:FOCdelta} implies that $\delta(t^*,s_l)+\E_f[c(\tilde{\theta})\mid s_l]$ is constant in $l$. Moreover, $\E_f[\delta(t^*,\tilde{s})]=0$ implies that this constant equals $\E_f[c(\tilde{\theta})]$. Hence, for all $l\in\{1,\ldots,m\}$,
	\[
	\delta(t^*,s_l)
	=
	\E_f[c(\tilde{\theta})]-\E_f[c(\tilde{\theta})\mid s_l].
	\]
	Using $c(\theta)=y(\theta)-\theta$, I obtain
	\[
	\E_{g^*}[y(\tilde{\theta})\mid s_l]
	=
	\E_f[c(\tilde{\theta})]+\E_f[\tilde{\theta}\mid s_l],
	\]
	i.e., the optimal agent behaves as if the conflict of interest were constant.
	
	Substituting the expression for $\delta(t,s_l)$ into $\delta(t^*,s_l)=\E_f[c(\tilde{\theta})]-\E_f[c(\tilde{\theta})\mid s_l]$ yields, for each $l$,
	\begin{align}\label{eq:deltaeq}
		\frac{1}{f_S(s_l)}\sum_{i=1}^{n-1}\bigl[y(\theta_{i+1})-y(\theta_i)\bigr]\bigl[t^*_{i,l-1}-t^*_{il}\bigr]
		=
		\E_f[c(\tilde{\theta})]-\E_f[c(\tilde{\theta})\mid s_l].
	\end{align}
	Iterating \autoref{eq:deltaeq} gives, for all $l\in\{1,\ldots,m-1\}$,
	\begin{align}\label{eq:sumconstraint}
		\sum_{i=1}^{n-1}\bigl[y(\theta_{i+1})-y(\theta_i)\bigr]t^*_{il}
		=
		\sum_{j=1}^{l} f_S(s_j)\Bigl(\E_f[c(\tilde{\theta})\mid s_j]-\E_f[c(\tilde{\theta})]\Bigr).
	\end{align}
	
	Now consider the matrix of transformations $t^\phi\in\R^{(n-1)\times(m-1)}$ with typical entry
	\[
	t^\phi_{kl}
	=
	\frac{\sum_{j=1}^{l} f_S(s_j)\bigl(\E_f[c(\tilde{\theta})\mid s_j]-\E_f[c(\tilde{\theta})]\bigr)}{y(\theta_{k+1})-y(\theta_k)}\,\phi_k,
	\]
	where $\phi=(\phi_1,\ldots,\phi_{n-1})\in\Delta^{n-1}$.\footnote{$\Delta^{a}$ denotes the $a$-dimensional simplex, i.e., the set of $a$-vectors satisfying $\phi_i\geq 0$ and $\sum_{i=1}^{a}\phi_i=1$.}
	By construction, $t^\phi$ satisfies \autoref{eq:sumconstraint}. Under the assumptions that $\E_f[c(\tilde{\theta})\mid s]$ is strictly decreasing in $s$ and that $y(\cdot)$ is strictly increasing, we have $t^\phi_{kl}\geq 0$. Thus, if $t^\phi\in\mathcal{T}$ for some $\phi\in\Delta^{n-1}$, then it solves the belief design problem and I conclude that there exists an overconfident optimal agent. By uniqueness of the induced $\delta(\cdot)$, any optimal agent acts as such.
	
	Let $\alpha:=\max_{s\in S}\bigl|\E_f[c(\tilde{\theta})\mid s]-\E_f[c(\tilde{\theta})]\bigr|$. Then
	\[
	|t^\phi_{kl}|
	=
	\left|
	\frac{\sum_{j=1}^{l} f_S(s_j)\bigl(\E_f[c(\tilde{\theta})\mid s_j]-\E_f[c(\tilde{\theta})]\bigr)}{y(\theta_{k+1})-y(\theta_k)}\,\phi_k
	\right|
	\leq
	\alpha\,\frac{\phi_k}{y(\theta_{k+1})-y(\theta_k)}.
	\]
	It follows that for all $i\in\{2,\ldots,n-1\}$ (and similarly at the boundaries using the convention $t_{0j}=t_{nj}=0$),
	\begin{align*}
		t^\phi_{i-1,j-1}-t^\phi_{i-1,j}-t^\phi_{i,j-1}+t^\phi_{ij}
		\geq
		-t^\phi_{i-1,j}-t^\phi_{i,j-1}
		\geq
		-\alpha\left[\frac{\phi_i}{y(\theta_{i+1})-y(\theta_i)}+\frac{\phi_{i-1}}{y(\theta_i)-y(\theta_{i-1})}\right]
	\end{align*}
	and
	\begin{align*}
		t^\phi_{i-1,j-1}-t^\phi_{i-1,j}-t^\phi_{i,j-1}+t^\phi_{ij}
		\leq
		t^\phi_{i-1,j-1}+t^\phi_{ij}
		\leq
		\alpha\left[\frac{\phi_i}{y(\theta_{i+1})-y(\theta_i)}+\frac{\phi_{i-1}}{y(\theta_i)-y(\theta_{i-1})}\right].
	\end{align*}
	Therefore, for each $\phi\in\Delta^{n-1}$ there exists $\bar{\alpha}(\phi)>0$ such that $\alpha<\bar{\alpha}(\phi)$ implies $t^\phi\in\mathcal{T}$. Let $\bar{\alpha}:=\sup_{\phi\in\Delta^{n-1}}\bar{\alpha}(\phi)>0$. Hence, whenever $\alpha<\bar{\alpha}$ there exists some $\phi\in\Delta^{n-1}$ such that $t^\phi\in\mathcal{T}$, which proves the first part of the proposition.
	
	The second part follows by symmetry: if
	\[
	\E_f[y(\tilde{\theta})\mid s']-\E_f[y(\tilde{\theta})\mid s]
	>
	\E_f[\tilde{\theta}\mid s']-\E_f[\tilde{\theta}\mid s]
	\quad \text{for all } s'>s,
	\]
	then the same construction with $t^\phi\leq 0$ yields an optimal transformation corresponding to an underconfident agent; by uniqueness of the induced $\delta(\cdot)$, any optimal agent acts as such.
	\end{proof}

	\begin{proof}[Proof of \autoref{prop:transfers}]
	The recommended actions must be incentive compatible given promised transfers and the agent's level of confidence. After each signal realization, two deviations are relevant: $(i)$ deviation to the other recommended action, and $(ii)$ deviation to the best action among those that are not recommended. The best deviation to an action yielding no transfer is to choose $\E_\tau[y(\tilde{\theta})\mid s_i]$ after signal $s_i$, which yields expected payoff $-Var_\tau(y(\tilde{\theta})\mid s_i)$.
		
		Let $(x_1,x_2)$ be the recommended actions and $(w_1,w_2)$ the associated transfers (paid only when the corresponding action is chosen). The four incentive compatibility constraints are
		\begin{align*}
			-\E_\tau[(x_1-y(\tilde{\theta}))^2\mid s_1]+w_1
			&\geq
			-\E_\tau[(x_2-y(\tilde{\theta}))^2\mid s_1]+w_2,\\
			-\E_\tau[(x_1-y(\tilde{\theta}))^2\mid s_1]+w_1
			&\geq
			-Var_\tau(y(\tilde{\theta})\mid s_1),\\
			-\E_\tau[(x_2-y(\tilde{\theta}))^2\mid s_2]+w_2
			&\geq
			-\E_\tau[(x_1-y(\tilde{\theta}))^2\mid s_2]+w_1,\\
			-\E_\tau[(x_2-y(\tilde{\theta}))^2\mid s_2]+w_2
			&\geq
			-Var_\tau(y(\tilde{\theta})\mid s_2).
		\end{align*}
		The principal's problem is therefore
		\[
		\max_{\{x_1,x_2,w_1,w_2,\tau\}}
		-\sum_{i=1}^2 f_S(s_i)\Bigl(\E_f[(x_i-\tilde{\theta})^2\mid s_i]+w_i\Bigr),
		\]
		subject to the above constraints.
		
		Let $\mu_i:=\E_\tau[y(\tilde{\theta})\mid s_i]$ and $\bar\mu := (\mu_1+\mu_2)/2$. Using
		\[
		\E_\tau[(x-y(\tilde{\theta}))^2\mid s_i]
		=
		Var_\tau(y(\tilde{\theta})\mid s_i) + (x-\mu_i)^2,
		\]
		the constraints can be rewritten as
		\begin{align}
			w_2-w_1
			&\le
			x_2^2-x_1^2 - 2(x_2-x_1)\mu_1,\label{eq:ic1}\\
			w_1
			&\ge
			(x_1-\mu_1)^2,\label{eq:ic2}\\
			w_2-w_1
			&\ge
			x_2^2-x_1^2 - 2(x_2-x_1)\mu_2,\label{eq:ic3}\\
			w_2
			&\ge
			(x_2-\mu_2)^2.\label{eq:ic4}
		\end{align}
		
		\medskip
		
		\noindent\textbf{Step 1 (monotonicity of recommended actions).}
		A necessary condition for \eqref{eq:ic1} and \eqref{eq:ic3} to be simultaneously satisfied is
		\[
		(x_2-x_1)(\mu_2-\mu_1)\ge 0.
		\]
		In any optimum $x_2\ge x_1$ must hold: otherwise the principal can replace $(x_1,x_2)$ by a constant action (e.g., their average), weakly reducing the expected losses by convexity, and (because the two recommended actions collapse) weakly relax the incentive constraints, allowing weakly lower transfers. Hence, in an optimum $x_2\ge x_1$ and therefore we also require $\mu_2\ge \mu_1$.
		
		\medskip
		
		\noindent\textbf{Step 2 (binding constraints).}
		At an optimum, at least one of \eqref{eq:ic2} and \eqref{eq:ic4} must bind; otherwise the principal can reduce both $w_1$ and $w_2$ by a small common amount without violating \eqref{eq:ic1}--\eqref{eq:ic4}. Moreover, at least two constraints bind: if only one binds, one can lower one of the wages without affecting feasibility.
		
		The binding pattern depends on the location of $(x_1,x_2)$ relative to $\bar\mu$. The three cases are:
\begin{table}[h]
	\centering
	\renewcommand{\arraystretch}{1.15}
	\setlength{\tabcolsep}{6pt}
	
	\begin{threeparttable}
		\resizebox{\textwidth}{!}{%
			\begin{tabular}{lccc}
				\toprule
				Case & 1 & 2 & 3 \\
				\midrule
				\makecell[l]{Binding\\constraints}
				& \makecell{\ref{eq:ic2} and \ref{eq:ic3}}
				& \makecell{\ref{eq:ic2} and \ref{eq:ic4}}
				& \makecell{\ref{eq:ic1} and \ref{eq:ic4}} \\
				\addlinespace[6pt]
				$w_1$
				& $\left(x_1-\mu_1\right)^2$
				& $\left(x_1-\mu_1\right)^2$
				& $\left(x_1-\mu_1\right)^2+2\left(\mu_2-\mu_1\right)\left(\bar{\mu}-x_2\right)$ \\
				\addlinespace[6pt]
				$w_2$
				& $\left(x_2-\mu_2\right)^2+2\left(\mu_2-\mu_1\right)\left(x_1-\bar{\mu}\right)$
				& $\left(x_2-\mu_2\right)^2$
				& $\left(x_2-\mu_2\right)^2$ \\
				\addlinespace[6pt]
				Restriction
				& $x_2\geq x_1>\bar{\mu}$
				& $x_1\leq \bar{\mu}\leq x_2$
				& $x_1\leq x_2<\bar{\mu}$ \\
				\bottomrule
		\end{tabular}}
	\end{threeparttable}
\end{table}

		\medskip
		
		Since beliefs enter the objective only through the induced wages, fix $(x_1,x_2)$ and define expected transfers
		\[
		W(x_1,x_2,\mu_1):=f_S(s_1)w_1+f_S(s_2)w_2.
		\]
		I first choose beliefs (equivalently, $(\mu_1,\mu_2)$) to minimize $W$ for given $(x_1,x_2)$, and then optimize over $(x_1,x_2)$.
		
		\medskip
		\noindent\textbf{Remark (Restriction to Case 2).}
		The analysis below focuses on Case~2, in which the average of the agent’s conditional preferred actions lies between the recommended actions. This restriction is without loss for Proposition~4. Cases~1 and~3 arise only when the average conflict of interest is sufficiently large in absolute value, in which case implementing the optimal action profile requires transfers that vary with the signal realization. Under the maintained assumption
		\[
		\bigl|\E_f[c(\tilde{\theta})]\bigr|
		\le
		\E_f[\tilde{\theta}\mid s_2]-\E_f[\tilde{\theta}\mid s_1],
		\]
		those cases are ruled out ex ante: the incentive constraints can be satisfied with flat transfers, and the optimal solution necessarily lies in Case~2. Since Proposition~4 is concerned precisely with this parameter region, it is sufficient to characterize the solution in Case~2.
		\medskip

		\medskip
		
		\noindent\textbf{Step 3 (Case 2 and the main solution).}
		Consider Case 2, i.e., $x_1\le \bar\mu\le x_2$, so that \eqref{eq:ic2} and \eqref{eq:ic4} bind and
		\[
		W(x_1,x_2,\mu_1)=f_S(s_1)(x_1-\mu_1)^2+f_S(s_2)(x_2-\mu_2)^2.
		\]
		Because the marginals of the joint distribution are fixed, $\E_f[y(\tilde{\theta})]$ is fixed. Hence
		\[
		\E_f[y(\tilde{\theta})]=f_S(s_1)\mu_1+f_S(s_2)\mu_2,
		\qquad\text{so}\qquad
		\mu_2=\frac{\E_f[y(\tilde{\theta})]-f_S(s_1)\mu_1}{f_S(s_2)}.
		\]
		Substituting into $W$ and differentiating yields
		\begin{align*}
			\frac{\partial W(x_1,x_2,\mu_1)}{\partial \mu_1}
			&=
			2f_S(s_1)\Bigl[x_2-x_1-(\mu_2-\mu_1)\Bigr],\\
			\frac{\partial^2 W(x_1,x_2,\mu_1)}{\partial \mu_1^2}
			&=
			2\frac{f_S(s_1)}{f_S(s_2)} \;>\;0.
		\end{align*}
		Therefore the unique minimizer satisfies
		\[
		\mu_2^*-\mu_1^* = x_2-x_1,
		\qquad\text{equivalently}\qquad
		\mu_1^*=\E_f[y(\tilde{\theta})]-f_S(s_2)(x_2-x_1).
		\]
		
		Now consider the principal's choice of $(x_1,x_2)$ given the wage-minimizing beliefs. By the envelope theorem,
		\[
		\frac{\partial W(x_1,x_2,\mu_1^*)}{\partial x_1}=2f_S(s_1)(x_1-\mu_1^*),
		\qquad
		\frac{\partial W(x_1,x_2,\mu_1^*)}{\partial x_2}=2f_S(s_2)(x_2-\mu_2^*).
		\]
		The first-order conditions for the principal's problem are
		\[
		-2f_S(s_i)\bigl(x_i-\E_f[\tilde{\theta}\mid s_i]\bigr)
		=
		\frac{\partial W(x_1,x_2,\mu_1^*)}{\partial x_i},
		\qquad i\in\{1,2\},
		\]
		which simplify to
		\[
		-(x_i-\E_f[\tilde{\theta}\mid s_i]) = x_i-\mu_i^*,
		\qquad i\in\{1,2\}.
		\]
		Thus,
		\[
		x_i^*=\frac{\mu_i^*+\E_f[\tilde{\theta}\mid s_i]}{2},
		\qquad i\in\{1,2\}.
		\]
		Using $\mu_i^*=\E_f[y(\tilde{\theta})\mid s_i]$ and $c(\theta)=y(\theta)-\theta$, I obtain
		\[
		\mu_i^*=\E_f[c(\tilde{\theta})]+\E_f[\tilde{\theta}\mid s_i],
		\qquad
		x_i^*=\frac{\E_f[c(\tilde{\theta})]}{2}+\E_f[\tilde{\theta}\mid s_i],
		\qquad i\in\{1,2\}.
		\]
		Moreover, the Case 2 restriction $x_1^*\le \bar\mu^*\le x_2^*$ is equivalent to
		\[
		\left|\E_f[c(\tilde{\theta})]\right|
		\le
		\E_f[\tilde{\theta}\mid s_2]-\E_f[\tilde{\theta}\mid s_1],
		\]
		which is exactly the assumption in the proposition. Hence, under the stated condition, the optimum indeed falls in Case 2.
		
		\medskip
		
		\noindent\textbf{Step 4 (concavity/verification).}
		After substituting the wage-minimizing beliefs, the principal's objective is a quadratic function of $(x_1,x_2)$. In Case 2, the second derivatives of $W$ are
		\begin{align*}
				\frac{\partial^2 W}{\partial x_1^2}&=2f_S(s_1)\bigl(1-f_S(s_2)\bigr)=2f_S(s_1)^2,\\
				\frac{\partial^2 W}{\partial x_2^2}&=2f_S(s_2)\bigl(1-f_S(s_1)\bigr)=2f_S(s_2)^2,\\
				\frac{\partial^2 W}{\partial x_1\partial x_2}&=2f_S(s_1)f_S(s_2).
		\end{align*}
	
		Therefore the Hessian of the objective is negative definite, and the solution above is the unique maximizer.
		
		\medskip
		
		\noindent Finally, since \eqref{eq:ic2} and \eqref{eq:ic4} bind in Case 2, the optimal wages are
		\[
		w_1^*=(x_1^*-\mu_1^*)^2, \qquad w_2^*=(x_2^*-\mu_2^*)^2,
		\]
		and with the expressions above one has $x_i^*-\mu_i^*=-\E_f[c(\tilde{\theta})]/2$, hence
		\[
		w_1^*=w_2^*=\frac{1}{4}\E_f[c(\tilde{\theta})]^2.
		\]
			
	\end{proof}
	
	Before proving \autoref{prop:delegation}, I introduce the following lemma.
\begin{lemma}
	\label{lemma:vars}
	Let $Var_f(\tilde{s}):=Var_f(\E_f[\tilde{\theta}\mid \tilde{s}])$. Then
	\[
	Var_f(\tilde{s})=\frac{\abs{f}^2}{f_S(s_1)f_S(s_2)}(\theta_2-\theta_1)^2.
	\]
\end{lemma}

	\begin{proof}[Proof of \autoref{lemma:vars}]
		Since $\tilde{s}$ is binary, we have
		\[
		Var_f(\tilde{s})
		=
		Var_f\!\big(\E_f[\tilde{\theta}\mid \tilde{s}]\big)
		=
		f_S(s_1)f_S(s_2)\bigl(\E_f[\tilde{\theta}\mid s_2]-\E_f[\tilde{\theta}\mid s_1]\bigr)^2.
		\]
		A direct calculation yields
		\[
		\E_f[\tilde{\theta}\mid s_2]-\E_f[\tilde{\theta}\mid s_1]
		=
		\frac{\abs{f}}{f_S(s_1)f_S(s_2)}(\theta_2-\theta_1),
		\]
		which implies
		\[
		Var_f(\tilde{s})
		=
		\frac{\abs{f}^2}{f_S(s_1)f_S(s_2)}(\theta_2-\theta_1)^2.
		\]
	\end{proof}

	\begin{proof}[Proof of \autoref{prop:delegation}]
	
		Delegation to an agent with beliefs $g^*$ is optimal if and only if
		\begin{align*}
			U(g^*)=-Var_f(\E_{g^*}[y(\tilde{\theta})|\tilde{s}] -\E_f[\tilde{\theta}|\tilde{s}])-\E_f[c(\tilde{\theta})]^2-\E_f[Var_f(\tilde{\theta}|\tilde{s}) ]\geq -Var_f(\tilde{\theta}).
		\end{align*}
		In any interior optimum for the belief choice we have
		\[
		Var_f(\E_{g^*}[y(\tilde{\theta})\mid \tilde{s}] - \E_f[\tilde{\theta}\mid \tilde{s}])=0,
		\]
		since otherwise a marginal adjustment of beliefs would strictly increase the principal’s payoff.  Therefore, this condition is equivalent to
		\begin{align*}
			\E_f[c(\tilde{\theta})]^2\leq Var_f(\tilde{\theta})-\E_f[Var_f(\tilde{\theta}|\tilde{s}) ]=Var_f(\tilde{s}).
		\end{align*}
		Using Lemma~\ref{lemma:vars}, this condition becomes
		\[
		\E_f[c(\tilde{\theta})]^2
		\le
		\frac{\abs{f}^2}{f_S(s_1)f_S(s_2)}(\theta_2-\theta_1)^2.
		\]
		Rearranging yields
		\[
		\frac{f_S(s_1)f_S(s_2)\E_f[c(\tilde{\theta})]^2}
		{\abs{f}(y(\theta_2)-y(\theta_1))(\theta_2-\theta_1)}
		-
		\abs{f}
		\le
		\Bigl(\frac{\theta_2-\theta_1}{y(\theta_2)-y(\theta_1)}-1\Bigr)\abs{f}
		=
		\tau^*,
		\]
		where $\tau^*$ is defined in the proof of \autoref{prop:2x2}.
	\end{proof}
	
	\section{Truth-or-noise information structure}
	\label{sec:truth}
	
	In this appendix, I analyze the baseline model from \autoref{sec:model} under an alternative
	information structure that is widely used in applied work. The goal is to
	illustrate the main mechanism of the paper in a continuous-state environment
	where confidence can be captured by a single parameter.
	
	\subsection*{Setting}
	
	The state $\theta$ lies in a compact interval
	$\Theta=[\underline{\theta},\bar{\theta}]\subset\mathbb{R}$. The agent observes a
	real-valued signal $s\in\mathbb{R}$. With probability $\rho\in(0,1)$, the signal
	perfectly reveals the state, so that $s=\theta$. With complementary probability
	$1-\rho$, the signal is uninformative and equals an independent draw from the
	state distribution.
	
	Let $F_\Theta$ denote the marginal distribution of the state, which is assumed
	to have full support and a strictly positive density. The joint distribution of
	$(\tilde{\theta},\tilde{s})$ induced by this information structure is denoted by
	$F$.
	
	The parameter $\rho$ represents the \textit{true} precision of the signal. The agent may
	misperceive this precision. Specifically, he acts \textit{as if} the signal reveals the
	state with probability $\rho+\kappa$, where
	$\kappa\in[-\rho,1-\rho]$. I refer to $\kappa$ as the agent's \textit{level of confidence}.
	When $\kappa>0$ the agent places more weight on the signal than is warranted,
	when $\kappa<0$ he places less weight on the signal, and when $\kappa=0$ he is
	well-calibrated.
	
	As in \autoref{sec:model}, the principal can select the agent based on his beliefs. The
	timing is as follows:
	\begin{enumerate}
		\item The principal chooses an agent with a level of confidence $\kappa$.
		\item Nature draws $(\tilde{\theta},\tilde{s})$ according to $F$.
		\item The agent observes $\tilde{s}$ and chooses an action $x\in\mathbb{R}$.
	\end{enumerate}
	
	Ex-post payoffs are quadratic. The principal's payoff is $-(x-\theta)^2$ and the agent's
	payoff is $-(x-y(\theta))^2$. Define the conflict of interest by
	$c(\theta):=y(\theta)-\theta$.
	
	\subsection*{Optimal level of confidence}
	
	The next result characterizes the principal's optimal choice of confidence in
	this environment.
	
	\begin{proposition}
		\label{prop:truth-noise}
		The unique optimal agent is:
		\begin{enumerate}[(i)]
			\item underconfident if and only if $y(\tilde{\theta})$ and $c(\tilde{\theta})$
			are positively correlated;
			\item well-calibrated if and only if $y(\tilde{\theta})$ and $c(\tilde{\theta})$
			are uncorrelated;
			\item overconfident if and only if $y(\tilde{\theta})$ and $c(\tilde{\theta})$
			are negatively correlated.
		\end{enumerate}
	\end{proposition}
	
	To interpret this result, define
	\[
	\beta :=
	\frac{Cov_F\!\left(y(\tilde{\theta}),c(\tilde{\theta})\right)}
	{Var_F\!\left(y(\tilde{\theta})\right)}.
	\]
	The coefficient $\beta$ measures how the conflict of interest varies with the
	agent's preferred action. When $\beta<0$, higher states are associated with a
	smaller conflict, implying that the agent's preferred action responds less to
	information than the principal's preferred action. In that case, increasing the
	agent's confidence improves incentives by amplifying his response to the signal.
	
	The optimal level of confidence is given by
	\[
	\kappa^*(\beta)=
	\begin{cases}
		-\rho & \text{if } \beta\geq 1,\\
		-\rho\,\beta & \text{if } \beta\in\left(-(1-\rho)/\rho,\,1\right),\\
		1-\rho & \text{if } \beta\leq -(1-\rho)/\rho.
	\end{cases}
	\]
	
	\begin{proof}[Proof of \autoref{prop:truth-noise}]
		Fix a level of confidence $\kappa$. Given a signal realization $s$, the agent
		believes that $s=\theta$ with probability $\rho+\kappa$. His optimal action therefore equals
		\[
		x_\kappa(s)
		=\E_\kappa[y(\tilde{\theta})\mid s]
		=(\rho+\kappa)y(s)+(1-\rho-\kappa)\E_F[y(\tilde{\theta})].
		\]
		
		The principal's expected payoff is
		\[
		U(\kappa)
		:=-\E_F\!\left[(x_\kappa(\tilde{s})-\tilde{\theta})^2\right]=-\E_F[x_\kappa(\tilde{s})^2]
		+2\E_F[x_\kappa(\tilde{s})\tilde{\theta}]
		-\E_F[\tilde{\theta}^2],
		\]
		where the last equality follows from expanding the square.
		
		Since $\tilde{s}$ has the same marginal distribution as $\tilde{\theta}$, we
		have
		\[
		\E_F[x_\kappa(\tilde{s})^2]
		=(\rho+\kappa)^2Var_F(y(\tilde{\theta}))
		+\E_F[y(\tilde{\theta})]^2.
		\]
		Moreover,
		\[
		\E_F[x_\kappa(\tilde{s})\tilde{\theta}]
		=(\rho+\kappa)\E_F[y(\tilde{s})\tilde{\theta}]
		+(1-\rho-\kappa)\E_F[y(\tilde{\theta})]\E_F[\tilde{\theta}].
		\]
		
		By the law of iterated expectations,
		\[
		\E_F[y(\tilde{s})\tilde{\theta}]
		=\rho\, Cov_F(y(\tilde{\theta}),\tilde{\theta})
		+\E_F[y(\tilde{\theta})]\E_F[\tilde{\theta}].
		\]

		Differentiating $U(\kappa)$ with respect to $\kappa$ yields
		\[
		\frac{\partial U(\kappa)}{\partial\kappa}
		=-2\kappa\,Var_F(y(\tilde{\theta}))
		+2\rho\!\left[Cov_F(y(\tilde{\theta}),\tilde{\theta})
		-Var_F(y(\tilde{\theta}))\right].
		\]
		This first-order condition reflects the tradeoff between increasing responsiveness to information and amplifying the agent's bias. Furthermore, the second derivative satisfies
		\[
		\frac{\partial^2 U(\kappa)}{\partial\kappa^2}
		=-2Var_F(y(\tilde{\theta}))<0,
		\]
		so the objective is strictly concave. Hence, the unique unconstrained maximizer solves
		\[
		\kappa
		=\rho\!\left(
		\frac{Cov_F(y(\tilde{\theta}),\tilde{\theta})}
		{Var_F(y(\tilde{\theta}))}
		-1\right)
		=-\rho
		\frac{Cov_F(y(\tilde{\theta}),c(\tilde{\theta}))}
		{Var_F(y(\tilde{\theta}))}
		=-\rho\beta.
		\]
		Imposing the constraint $\kappa\in[-\rho,1-\rho]$ yields the stated expression for
		$\kappa^*(\beta)$.
	\end{proof}
	
\subsection*{Transfers}
Parallel to \autoref{sec:transfers}, I allow the principal to commit \textit{ex ante} to a
nonnegative wage schedule $w:\mathbb{R}\to\mathbb{R}_+$ that conditions on the
agent’s action, in addition to selecting the agent’s beliefs. I study how the
availability of such action-contingent transfers interacts with belief design
under a truth-or-noise information structure.

	\begin{proposition}
		\label{prop:transfers-truth}
		Under the truth-or-noise information structure, suppose that
		\[
		\E_F\!\bigl[c(\tilde{\theta})\bigr]^2
		\le
		Var_F\!\bigl(y(\tilde{\theta})\bigr).
		\]
		Then, the availability of action-contingent transfers does not affect the
		principal’s optimal choice of confidence (so that \autoref{prop:truth-noise}
		continues to apply). Moreover, in the optimum, transfers are constant across
		equilibrium actions.
	\end{proposition}
	
	The conclusion mirrors that of \autoref{sec:transfers}: even when transfers are available, optimal belief design continues to govern how the agent’s actions respond to information, while transfers serve only to correct the average bias.
	
\begin{proof}[Proof of \autoref{prop:transfers-truth}]
	Fix a level of confidence $\kappa\in[-\rho,1-\rho]$. The principal chooses an
	action recommendation rule $x:\mathbb{R}\to\mathbb{R}$ and a nonnegative wage
	schedule $w:\mathbb{R}\to\mathbb{R}_+$. After observing $s$, the agent selects
	$x\in\mathbb{R}$ to maximize
	\[
	-\E_\kappa\!\left[(x-y(\tilde{\theta}))^2\mid s\right]+w(x).
	\]
	
	Let $\mu(s):=\E_\kappa[y(\tilde{\theta})\mid s]$ denote the agent’s conditional
	preferred action. Under the truth-or-noise information structure,
	\[
	\mu(s)=(\rho+\kappa)y(s)+(1-\rho-\kappa)\E_F[y(\tilde{\theta})].
	\]
	Moreover, conditional on $s$, with probability $\rho+\kappa$ we have
	$\tilde{\theta}=s$, while with probability $1-\rho-\kappa$ the state is an
	independent draw from $F_\Theta$. Hence,
	\begin{equation}
		\label{eq:var-truth}
		Var_\kappa\!\bigl(y(\tilde{\theta})\mid s\bigr)
		=
		(1-\rho-\kappa)\,Var_F\!\bigl(y(\tilde{\theta})\bigr),
	\end{equation}
	which is constant across signal realizations.
	
	\medskip
	\noindent\textbf{Step 0 (restriction to off-path zero wages).}
	Let $X:=\{x(s): s\in\mathbb{R}\}$ denote the set of equilibrium (recommended)
	actions. Consider any wage schedule $w$ and define $\hat w$ by
	\[
	\hat w(x)=
	\begin{cases}
		w(x) & \text{if } x\in X,\\
		0 & \text{if } x\notin X.
	\end{cases}
	\]
	Since $\hat w(x)\le w(x)$ for all $x$, replacing $w$ by $\hat w$ weakly increases
	the principal’s payoff. Moreover, for each signal realization $s$, the wages
	associated with on-path actions are unchanged, while every off-path action yields
	weakly lower wages. Therefore, any incentive compatibility constraint that holds
	under $w$ also holds under $\hat w$. Hence, without loss of generality, restrict
	attention to wage schedules satisfying
	\[
	w(x)=0 \qquad \text{for all } x\notin X.
	\]
	
	\medskip
	\noindent\textbf{Step 1 (incentive compatibility constraints).}
	Fix a signal realization $s$. Under the restriction above, deviations fall into
	two classes.
	
	\smallskip
	\noindent (i) \emph{Deviations to other on-path actions.}
	For any $s'\in\mathbb{R}$, deviating to the action $x(s')\in X$ yields wage
	$w(x(s'))$. Incentive compatibility therefore requires that for all $s,s'$,
	\begin{equation}
		\label{eq:ic-onpath}
		-\E_\kappa\!\left[(x(s)-y(\tilde{\theta}))^2\mid s\right]+w(x(s))
		\;\ge\;
		-\E_\kappa\!\left[(x(s')-y(\tilde{\theta}))^2\mid s\right]+w(x(s')).
	\end{equation}
	
	\smallskip
	\noindent (ii) \emph{Deviations to off-path actions.}
	Any $x\notin X$ yields zero wage. Among such deviations, the agent chooses
	$\mu(s)$, which minimizes $\E_\kappa[(x-y(\tilde{\theta}))^2\mid s]$. Thus,
	incentive compatibility requires
	\begin{equation}
		\label{eq:ic-offpath}
		-\E_\kappa\!\left[(x(s)-y(\tilde{\theta}))^2\mid s\right]+w(x(s))
		\;\ge\;
		-Var_\kappa\!\bigl(y(\tilde{\theta})\mid s\bigr).
	\end{equation}
	
	Using
	\[
	\E_\kappa\!\left[(x-y(\tilde{\theta}))^2\mid s\right]
	=
	Var_\kappa\!\bigl(y(\tilde{\theta})\mid s\bigr)+(x-\mu(s))^2,
	\]
	constraint \eqref{eq:ic-offpath} is equivalent to
	\begin{equation}
		\label{eq:ic-offpath2}
		w(x(s)) \;\ge\; (x(s)-\mu(s))^2
		\qquad \text{for all } s.
	\end{equation}
	
	\medskip
	\noindent\textbf{Step 2 (binding of off-path constraints and flatness of transfers).}
	Since the principal’s objective is decreasing in wages, constraint
	\eqref{eq:ic-offpath2} binds pointwise in any optimum:
	\begin{equation}
		\label{eq:w-binding}
		w(x(s))=(x(s)-\mu(s))^2
		\qquad \text{for all } s.
	\end{equation}
	Therefore, transfers are constant across equilibrium actions if and only if
	$x(s)-\mu(s)$ is constant in $s$. In particular, if there exists $d\in\mathbb{R}$
	such that
	\begin{equation}
		\label{eq:constant-wedge}
		x(s)-\mu(s)=d
		\qquad \text{for all } s,
	\end{equation}
	then \eqref{eq:w-binding} implies
	\[
	w(x(s))=d^2
	\qquad \text{for all } s.
	\]
	Define $w^*:=d^2$.
	
	\medskip
	\noindent\textbf{Step 3 (optimal beliefs and separation from transfers).}
	For any fixed recommendation rule $x(\cdot)$, the incentive compatibility
	constraints \eqref{eq:ic-onpath}--\eqref{eq:ic-offpath2} depend on the agent’s
	beliefs only through the induced conditional expectations $\mu(s)$, since
	\eqref{eq:var-truth} implies that
	$Var_\kappa(y(\tilde{\theta})\mid s)$ is constant across signal realizations.
	
	Hence, whenever it is feasible to choose beliefs such that
	\eqref{eq:constant-wedge} holds, the principal optimally does so, eliminating
	signal-dependent distortions exactly as in the baseline model. Under the
	maintained assumption
	\[
	\E_F[c(\tilde{\theta})]^2 \le Var_F\!\bigl(y(\tilde{\theta})\bigr),
	\]
	this feasibility region includes the principal’s optimal action profile.
	Therefore, the optimal choice of confidence coincides with that in
	\autoref{prop:truth-noise}, and \eqref{eq:constant-wedge} holds in equilibrium.
	By Step~2, this implies $w(x)= w^*$ across recommended actions $x\in X$.
\end{proof}

	 \subsection*{Delegation}
As in \autoref{sec:delegation}, I study the conditions under which the principal prefers to delegate the decision to an agent with the optimal level of confidence.
	 
\begin{proposition}
	\label{prop:delegation-truth}
	Under the truth-or-noise information structure, the decision is delegated to the
	optimal agent if and only if
	\[
	\kappa^*
	\;\ge\;
	\frac{\abs{\E_F[c(\tilde{\theta})]}}
	{\sqrt{Var_F(y(\tilde{\theta}))}}
	-
	\rho,
	\]
	where $\kappa^*$ denotes the principal’s optimal choice of confidence.
\end{proposition}

Exactly as in \autoref{sec:delegation}, if the agent is on average unbiased (i.e., $\E_f[c(\tilde{\theta})]=0$), he acts identically to the principal under centralization. In this case, the principal is indifferent between delegating to a maximally underconfident agent and centralizing decision-making. By contrast, when $\E_f[c(\tilde{\theta})]\neq 0$,  delegation requires that the optimal agent’s level of confidence be strictly above  $-\rho$ (the confidence of a maximally underconfident agent).

	\begin{proof}[Proof of \autoref{prop:delegation-truth}]
		Under centralization, the principal does not condition on the agent’s signal and
		chooses the action $x=\E_F[\tilde{\theta}]$. Her expected payoff therefore equals $U^C:=-Var_F(\tilde{\theta})$.
			
		Under delegation to an agent with confidence $\kappa$, the agent’s action after
		observing $s$ is
		\[
		x_\kappa(s)
		=
		(\rho+\kappa)y(s)+(1-\rho-\kappa)\E_F[y(\tilde{\theta})].
		\]
		The principal’s expected payoff from delegation is
		\[
		U^{D}(\kappa)
		=
		-\E_F\!\left[(x_\kappa(\tilde{s})-\tilde{\theta})^2\right].
		\]
		
		Expanding the square and using the law of iterated expectations yields
		\[
		U^{D}(\kappa)
		:=
		-\E_F\!\bigl[c(\tilde{\theta})\bigr]^2
		-(1-\rho-\kappa)^2\,Var_F\!\bigl(y(\tilde{\theta})\bigr)
		-\rho(1-\rho)Var_F(\tilde{\theta}),
		\]
		where the last term captures the irreducible uncertainty due to the truth-or-noise
		structure.
		
		Delegation to the optimally chosen agent is optimal if and only if
		\[
		U^{D}(\kappa^*) \ge U^{C}.
		\]
		Substituting the expressions above and rearranging, this condition is equivalent to
		\[
		\E_F\!\bigl[c(\tilde{\theta})\bigr]^2
		\le
		(\rho+\kappa^*)^2\,Var_F\!\bigl(y(\tilde{\theta})\bigr).
		\]
		Rearranging yields the result.
	\end{proof}

\end{document}